\documentclass[10pt,journal,compsoc]{IEEEtran}

\ifCLASSOPTIONcompsoc
  \usepackage[nocompress]{cite}
\else
  \usepackage{cite}
\fi

\ifCLASSINFOpdf
   \usepackage[pdftex]{graphicx}
 
\else
 
\fi

\usepackage{fixltx2e}
\usepackage{amsmath,amsfonts,amssymb,amsthm}
\newtheorem{theorem}{Theorem}

\usepackage[noend]{algpseudocode}
\usepackage{algorithmicx,algorithm}

\usepackage{array}

\usepackage{subfigure} 

\usepackage{booktabs}
\usepackage{multirow}
\usepackage{makecell}
\usepackage{diagbox}
\usepackage{color}

\usepackage{url}

\begin{document}

\title{Reversible Quantization Index Modulation for Static Deep Neural Network Watermarking}

\author{Junren Qin, Shanxiang Lyu, Fan Yang, Jiarui Deng, Zhihua Xia, Xiaochun Cao
\IEEEcompsocitemizethanks{\IEEEcompsocthanksitem Junren Qin, Shanxiang Lyu, Fan Yang, Jiarui Deng and Zhihua Xia are with the College of Cyber Security, Jinan University, Guangzhou 510632, China (E-mail: lsx07@jnu.edu.cn).  Xiaochun Cao is with the School of Cyber Science and Technology, Sun Yatsen University, Shenzhen Campus, Shenzhen 518107, China (E-mail: caoxiaochun@mail.sysu.edu.cn).
}
}


\IEEEtitleabstractindextext{%

\begin{abstract}
	Static deep neural network (DNN) watermarking techniques typically employ irreversible methods to embed watermarks into the DNN model weights. However, this approach causes permanent damage to the watermarked model and fails to meet the requirements of integrity authentication. Reversible data hiding (RDH) methods offer a potential solution, but existing approaches suffer from weaknesses in terms of usability, capacity, and fidelity, hindering their practical adoption.
	In this paper, we propose a novel RDH-based static DNN watermarking scheme using quantization index modulation (QIM). Our scheme incorporates a novel approach based on a one-dimensional quantizer for watermark embedding. Furthermore, we design two schemes to address the challenges of integrity protection and legitimate authentication for DNNs. Through simulation results on training loss and classification accuracy, we demonstrate the feasibility and effectiveness of our proposed schemes, highlighting their superior adaptability compared to existing methods.
\end{abstract}

}

\maketitle
\IEEEdisplaynontitleabstractindextext
\IEEEpeerreviewmaketitle
 
\section{Introduction}\label{sec:introduction}
\IEEEPARstart{D}{eep} neural networks (DNNs) have gained significant popularity due to their remarkable performance and have found applications in various fields \cite{montavon2018methods,sze2017efficient,szegedy2013deep,samek2021explaining,radhakrishnan2023wide,tang2021steganography,lou2022112333}. However, the increasing use of deep learning-based systems also poses a risk of unauthorized usage or modification of DNN models without proper attribution to the original authors. To address this concern, watermarking techniques for DNNs have emerged as an important step in protecting the intellectual property embedded in these models \cite{barni2021challenge}. Watermarking provides an additional layer of security that allows the original authors to prove ownership of their models, safeguard them from unauthorized access and use, track their provenance, ensure integrity, facilitate versioning, and identify malicious models \cite{regazzoni2021protecting}.

Deep neural network (DNN) watermarking techniques can be broadly categorized into static and dynamic watermarking approaches \cite{li2021survey,zhang2018protecting}, depending on where the watermark can be read from. In static watermarking methods (e.g., \cite{uchida2017embedding, kuribayashi2021white, li2021spread, pagnotta2022tattooed}), the watermark can be directly extracted from the network weights. During the training phase, these weights are determined and typically represented in floating-point formats, which differ from the popular unsigned-integer format commonly used for images. Static watermarking techniques aim to embed the watermark directly into the weights of the DNN, ensuring that the ownership and integrity of the model can be verified by examining these weight values. Static watermarking is particularly relevant in scenarios where the protection of the model's weights and ownership verification are of utmost importance. On the other hand, dynamic watermarking techniques (e.g., \cite{li2021payloadwatermarking, fei2022GanWatermarking}) rely on the modification of the network's behavior when provided with specific inputs, resulting in a visible watermark in the model's output. By carefully designing the input signals or modifying the network's architecture, dynamic watermarking allows for the extraction of the watermark through the observation of specific output patterns. Dynamic watermarking techniques offer a more flexible approach by embedding the watermark in the network's behavior rather than its weights. This enables the watermark to be extracted from the model's output, making it suitable for applications where the focus is on detecting unauthorized usage or tracking the dissemination of the model.
 
An intriguing research direction is the development of reversible watermarking schemes for DNNs. Reversible watermarking is a type of digital watermarking that enables content owners to protect their digital data without causing any permanent modifications \cite{shi2016reversible, DBLP:journals/sigpro/HuaHSGT16}. It allows embedded information to be retrieved from the host object without any data loss or damage. Reversible watermarking algorithms have been successfully applied to the unsigned integer format commonly used in images, including techniques such as difference expansion (DE) \cite{1227616}, prediction-error expansion (PEE) \cite{1421361, WU2014387}, and histogram shifting (HS) \cite{ni2006reversible}. Considering that the weights in DNNs can be treated as conventional multimedia objects, reversible watermarking of DNNs can be seen as an extension of static DNN watermarking, where reversible watermarks are embedded within the weights. However, existing approaches, such as the HS-based method proposed by Guan \textit{et al.} \cite{Guan_2020} for watermarking convolutional neural networks (CNNs), face challenges when dealing with floating-point weights and suffer from degradation when the host exhibits a uniform or uniform-like distribution.

Motivated by these challenges, we propose a novel reversible watermarking scheme specifically tailored for floating-point weights in DNNs. Our contributions, along with their highlights, are summarized as follows:

\begin{itemize}
	\item First, we design a simple yet efficient reversible watermarking algorithm, named reversible quantization index modulation (R-QIM), which improves upon the widely used quantization index modulation (QIM) \cite{chen2001quantization,moulin2005data,lyu2023optimized,Feng2016}. R-QIM allows for reversible embedding of watermarks in floating-point or real-valued objects, resembling a lattice quantizer that maps input values from a large continuous set to a countable smaller set with a finite number of elements. While QIM is naturally lossy, we leverage the availability of the cover object during the watermark embedding process to add a scaled version of the difference vector back to the quantized output values, enabling reversibility.
	\item Second, we demonstrate how R-QIM can be deployed in DNN watermarking to achieve integrity protection and legitimacy authentication. For integrity protection, our scheme allows the owner or a trusted third-party institution to verify the occurrence of data tampering, regardless of noiseless or known noisy channel conditions. This addresses the limitations of existing schemes that are unavailable in noisy channel transmission. For legitimacy authentication, our proposed scheme provides an effective means to differentiate between legal and illegal use of target DNNs. This added layer of protection helps deter attackers and facilitates the identification of individuals responsible for unauthorized use. Additionally, it provides assurance that a given DNN is authentic, ensuring the integrity of the produced data.
	\item Third, we provide theoretical justifications and conduct numerical simulations to showcase the advantages of R-QIM. We analyze the signal-to-watermark ratio (SWR) of R-QIM, which measures capacity and fidelity, and compare the training loss and classification accuracy of R-QIM with the HS-based method \cite{Guan_2020} by analyzing the weights of multi-layer perceptron (MLP) and visual geometry group (VGG) models.
\end{itemize}

The remainder of the paper is organized as follows. Section \ref{Sec. related work} introduces DNN watermarking models and existing algorithms. Sections \ref{Sec. Mathmetical Model} and \ref{Sec. QY-RW scheme} present R-QIM along with theoretical analyses and its applications in DNN watermarking. Section \ref{Sec. experiment} provides simulation results, and Section \ref{Sec. conclusion} concludes the paper.

\section{Preliminaries}\label{Sec. related work}

\subsection{Reversible DNN Watermarking Basics}
Reversible deep neural network (DNN) watermarking involves the embedding of a watermark into the weights of a DNN model in a manner that allows for its extraction without any permanent modifications or loss of information. This reversible embedding process is analogous to static DNN watermarking, where the watermark is embedded directly into the network weights during the training phase \cite{uchida2017embedding,li2021survey}. However, reversible watermarking techniques ensure that the original weights can be perfectly recovered after the watermark is extracted.

The mathematical model for reversible DNN watermarking can be described as follows. Let $\mathbf{W}$ denote the set of all weights in a trained DNN model. During watermark embedding, specific weights from $\mathbf{W}$ are selected based on a location sequence $\mathbf{c}$ guided by a clue or key $cl$, resulting in a cover sequence $\mathbf{s}$. The information sequence $\mathbf{m}$ is then embedded into $\mathbf{s}$ using a carefully designed embedding function $\textrm{Emb}(\cdot)$, resulting in the watermarked sequence $\mathbf{s}_w$. 

To ensure correct extraction and recovery of the watermark, the following triplet of operations is applied:
\begin{equation}
	\left\{ \begin{aligned}
		\mathbf{s}_{w} &= \textrm{Emb}(\mathbf{s}, \mathbf{m}) \\
		\hat{\mathbf{m}} &= \textrm{Ext}(\mathbf{s}_w + \mathbf{n}) = \textrm{Ext}(\mathbf{y}) \\
		\hat{\mathbf{s}} &= \textrm{Rec}(\mathbf{s}_w + \mathbf{n}) = \textrm{Rec}(\mathbf{y}) \\
	\end{aligned}\right.
\end{equation}
where  $\textrm{Emb}(\cdot)$ represents the embedding function that embeds the information sequence $\mathbf{m}$ into the cover sequence $\mathbf{s}$ to produce the watermarked sequence $\mathbf{s}_w$. $\textrm{Ext}(\cdot)$ and $\textrm{Rec}(\cdot)$ denote the extraction and recovery functions, respectively. $\mathbf{n}$ represents the additive noise present in the received watermarked sequence $\mathbf{y} = \mathbf{s}_w + \mathbf{n}$.

While reversible DNN watermarking shares similarities with reversible image watermarking, there are notable differences in terms of the cover format, robustness, and fidelity requirements. Table \ref{Tab. different requirements} summarizes the key differences between reversible image watermarking and reversible DNN watermarking. Reversible DNN watermarking operates on floating-point weights, which differ from the unsigned integers typically used in reversible image watermarking. The fidelity requirement in reversible DNN watermarking pertains to the effectiveness of the host network after watermark embedding, rather than the visual quality of the host signal as in image watermarking. Additionally, reversible DNN watermarking should have the capacity to embed a large amount of data or information into the network weights. Security is crucial to prevent unauthorized parties from accessing, reading, or modifying the watermark. Lastly, efficiency is important to ensure faster embedding and extraction processes for DNN watermarking algorithms.

By understanding the unique characteristics and requirements of reversible DNN watermarking, we can develop tailored algorithms and techniques that enable the embedding, extraction, and recovery of watermarks while preserving the integrity and effectiveness of the DNN models.

\begin{table*}[t!]
	\centering
	\renewcommand\arraystretch{1.5}
	\caption{Comparison of reversible image watermarking and reversible DNN watermarking.}
	\label{Tab. different requirements}
	\begin{tabular}{p{.1\textwidth}|| p{.4\textwidth} p{.4\textwidth}}
		\hline 
		Features & Reversible image watermarking& Reversible DNN watermarking\\
		\hline		
		Format of covers & Unsigned integers & Floating-point numbers \\		
		Fidelity & Higher quality of the host signal after watermark embedding& Higher effectiveness of the host network after watermark embedding \\
		Capacity&\multicolumn{2}{p{.8\textwidth}}{Ability to embed a watermark with a
			massive amount of data/information}\\
		Security&\multicolumn{2}{p{.8\textwidth}}{Ability to remain secret from unauthorized parties accessing, reading, and modifying the watermark}\\
		Efficiency&\multicolumn{2}{p{.8\textwidth}}{Higher speed for the embedding and extraction process of the watermarking algorithm}\\
		\hline
	\end{tabular}
\end{table*}

\subsection{Existing Methods}
\subsubsection{HS}
HS (Histogram Shifting) is a reversible watermarking algorithm originally developed for images, but it has been adapted for use in CNNs \cite{Guan_2020}. The method consists of three main parts: host sequence construction, data preprocessing, and the watermarking algorithm.

In the host sequence construction, a host matrix is constructed from a convolutional layer in the CNN. This step is not directly relevant to this paper and will not be discussed further. In the data preprocessing step, each weight is defined as follows:

\begin{equation}
	\omega=\pm 0.\underbrace{00...0}_{p~\text{digits}}n_{1}n_{2}...n_{c}n_{c+1}...n_{q},
\end{equation}
Here, $q$ represents the total length of digits for the weight. To meet the requirements of an integer host, the consecutive non-zero digit pairs $(n_{c},n_{c+1})$ in $\omega$, corresponding to the minimum entropy, are chosen as the significant digit pairs to construct the host sequence. These chosen pairs are then adjusted by adding an adjustable integer parameter $V$ to ensure they fall within the appropriate range of $[-99,99]$.

For the watermarking algorithm, HS \cite{ni2006reversible} scheme is employed as the embedding and extraction strategy. The 1-bit HS embedding process for the watermark $m$ can be described as follows:

\begin{equation}
	\omega^{'}=\left\{
	\begin{aligned}
		&\omega+m,&\omega=\Omega_{\mathrm{max}}\\
		&\omega+1,&\omega\in\left(\Omega_{\mathrm{max}},\Omega_{\mathrm{min}}\right)\\
		&\omega,&\omega\notin\left[\Omega_{\mathrm{max}},\Omega_{\mathrm{min}}\right)\\
	\end{aligned}\right..
\end{equation}
The histogram shifting operation in this 1-bit embedding process is depicted in Fig. \ref{fig_HS_shift}(a), where the bins greater than $\Omega_{\mathrm{max}}$ are shifted to the right by a fixed $\Delta=1$ to create a vacant bin for embedding. The watermark $m$ with a uniform distribution is then embedded into the bin equal to $\Omega_{\mathrm{max}}$ using HS. This divides the entire cover into three regions, as depicted in Fig. \ref{fig_HS_shift}(c): region i for covers smaller than $\Omega_{\mathrm{max}}$, region ii for covers equal to $\Omega_{\mathrm{max}}$, and region iii for covers larger than $\Omega_{\mathrm{max}}$. The mapping rule for $\omega$ changes depending on the bit, as shown in Fig. \ref{fig_HS_shift}(b).

Using the same process of host sequence construction and data preprocessing, the extraction process can be described as follows:

\begin{equation}
	\hat{m}=\left\{
	\begin{aligned}
		&1,&\omega^{'}=\Omega_{\mathrm{max}}+1\\
		&0,&\omega^{'}=\Omega_{\mathrm{max}}\\
	\end{aligned}\right.,
\end{equation}
and the recovery process as:
\begin{equation}
	\hat{\omega}=\left\{
	\begin{aligned}
		&\omega^{'}-1,&\omega\in\left(\Omega_{\mathrm{max}},\Omega_{\mathrm{min}}\right)\\
		&\omega^{'},&\omega\notin\left[\Omega_{\mathrm{max}},\Omega_{\mathrm{min}}\right)\\
	\end{aligned}\right..
\end{equation}

\begin{figure}[t!]
	\centering
	\subfigure[Histogram shifting.]{\includegraphics[width=.42\textwidth]{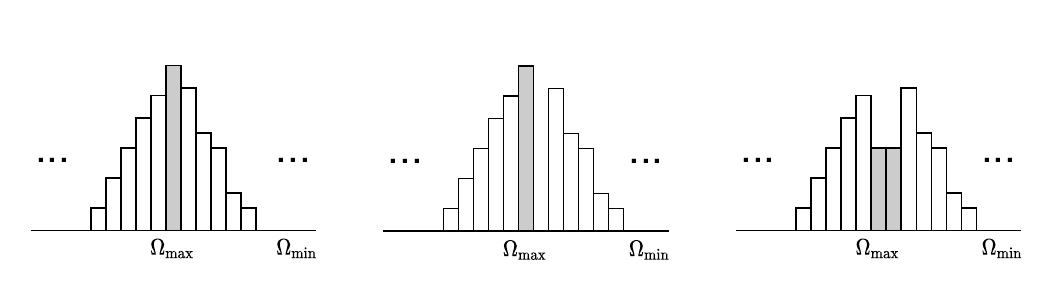}} 
	\subfigure[The mapping rule.]{\includegraphics[width=.42\textwidth]{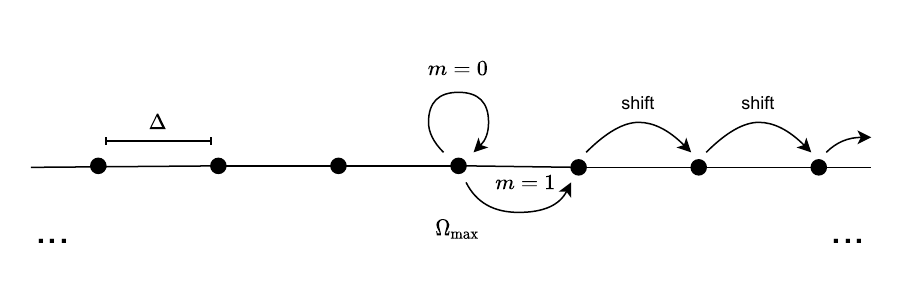}}
	\subfigure[The three partitioned regions.]{\includegraphics[width=.42\textwidth]{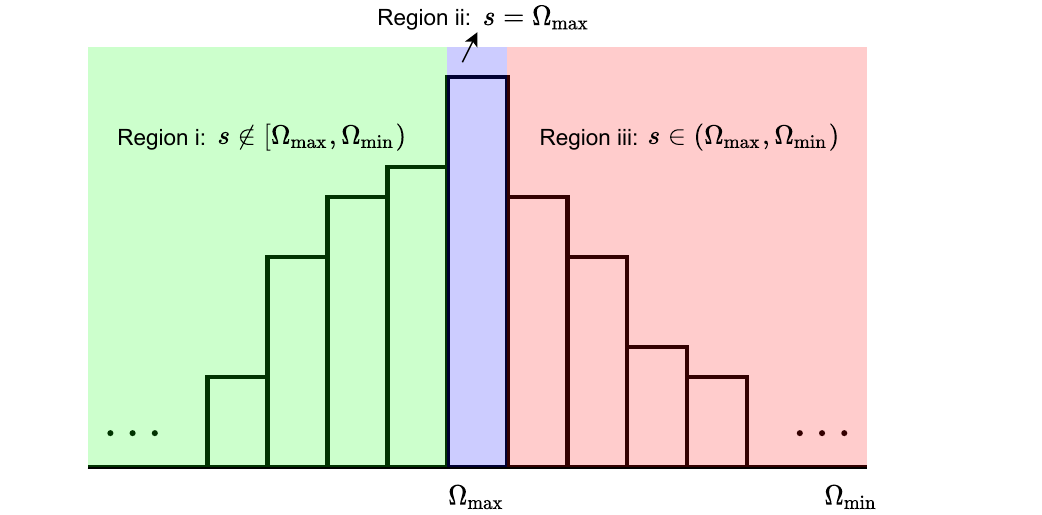}} 
	\caption{Illustration of the HS algorithm.} 
	\label{fig_HS_shift}
\end{figure}

\subsubsection{QIM}\label{Sec.QIM}
QIM (Quantization Index Modulation) is a widely used method for non-reversible watermarking \cite{chen2001quantization,moulin2005data,lyu2023optimized,Feng2016}. Its rationale can be explained using the example shown in Fig. \ref{Fig. QIM embed model}(a). The circle and cross positions in Fig. \ref{Fig. QIM embed model}(a) represent two sets, $\Lambda_0$ and $\Lambda_1$, arranged alternately. Given a host or cover sample $s\in\mathbb{R}$ and a one-bit message $m \in \{0,1\}$, the watermarked value is obtained by moving $s$ to the nearest point in $\Lambda_0$ when $m=0$, and to the nearest point in $\Lambda_1$ when $m=1$.

Let $Q_{\Delta}(s)=\Delta\lfloor s/\Delta\rfloor$ be a quantization function with $\Delta$ as the step-size parameter. The embedding process can be described as follows:
\begin{equation}\label{Eq. conventional QIM embedding}
	s_{\mathrm{QIM}}\triangleq Q_m(s)=Q_{\Delta}(s-d_m)+d_m,~m\in \{0,1\},
\end{equation}
where $d_0=-(\Delta/4)$, $d_1=\Delta/4$, $\Lambda_0=d_0+\Delta\mathbb{Z}$, and $\Lambda_1=d_1+\Delta\mathbb{Z}$.

Assuming that the transmitted $s_{\mathrm{QIM}}$ has been contaminated by an additive noise term $n$, the received signal is given by $y= s_{\mathrm{QIM}}+n$. A minimum distance decoder is used to extract the watermark as follows:

\begin{equation}\label{Eq. QIM extract}
	\hat{m}=\mathop{arg\min}_{m\in\{0,1\}}\left[\mathop{\min}_{s\in\Lambda_m}|y-s|\right].
\end{equation}
If $|n|<\Delta/4$, the estimated value $\hat{m}$ is correct.

In terms of embedding distortion, as shown in Fig. \ref{Fig. QIM distortion}(a), the maximum error caused by embedding is $\Delta/2$. If the quantization errors are uniformly distributed over $[-(\Delta/2),(\Delta/2)]$, the mean-squared embedding distortion is given by $D=\Delta^2/12$. Considering the capacity, QIM achieves an approximate rate of 1 bpps (bit per sample), which means each sample of the host cover can carry 1 bit of watermark information.

\begin{figure}[t!]
	\centering
	\includegraphics[width=.5\textwidth]{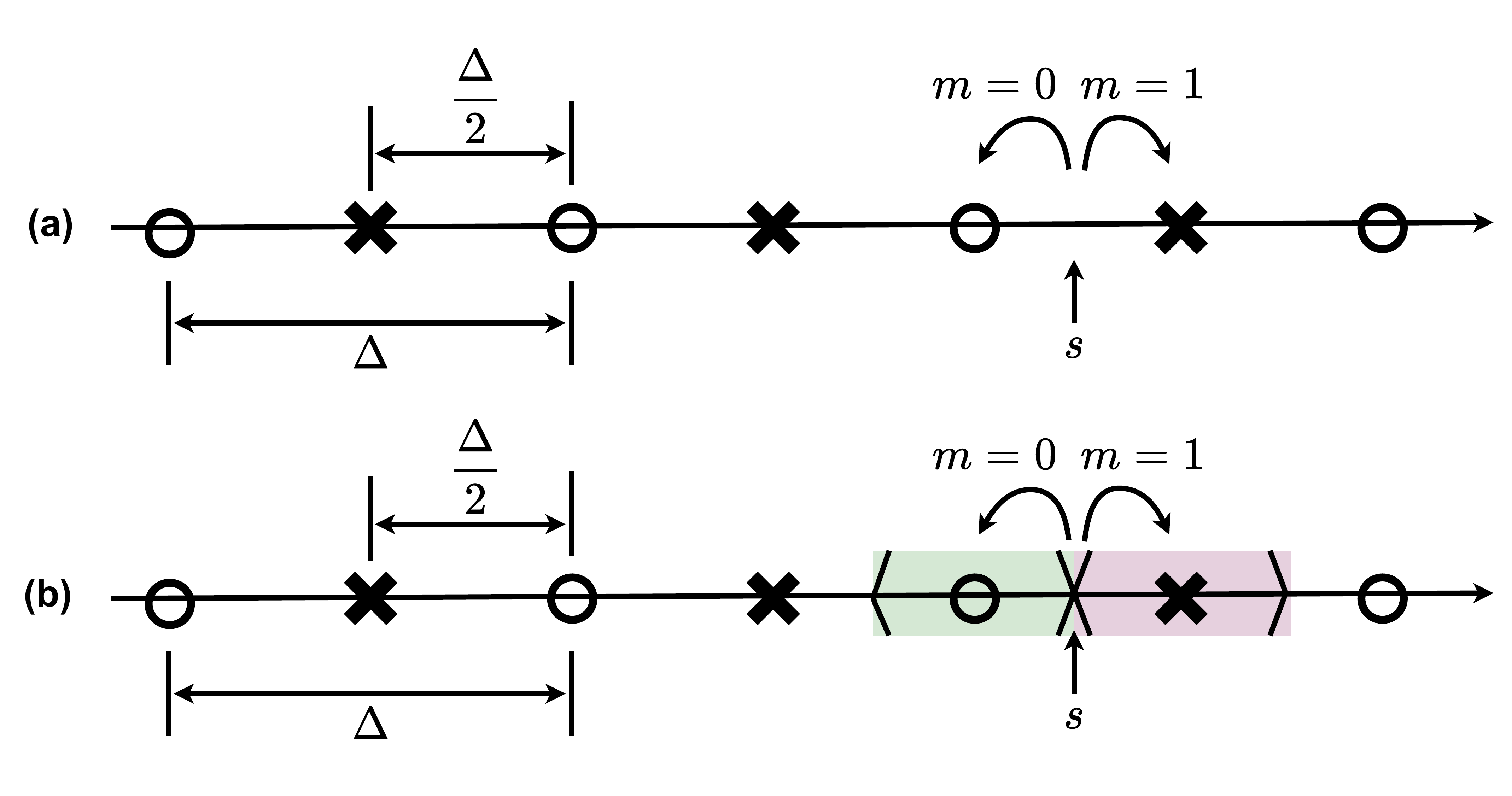} 
	\caption{Embed one bit into a sample with different versions of QIM. (a) Conventional QIM. (b) Reversible QIM.} 
	\label{Fig. QIM embed model}
\end{figure}

\begin{figure}[t!]
	\centering
	\includegraphics[width=.42\textwidth]{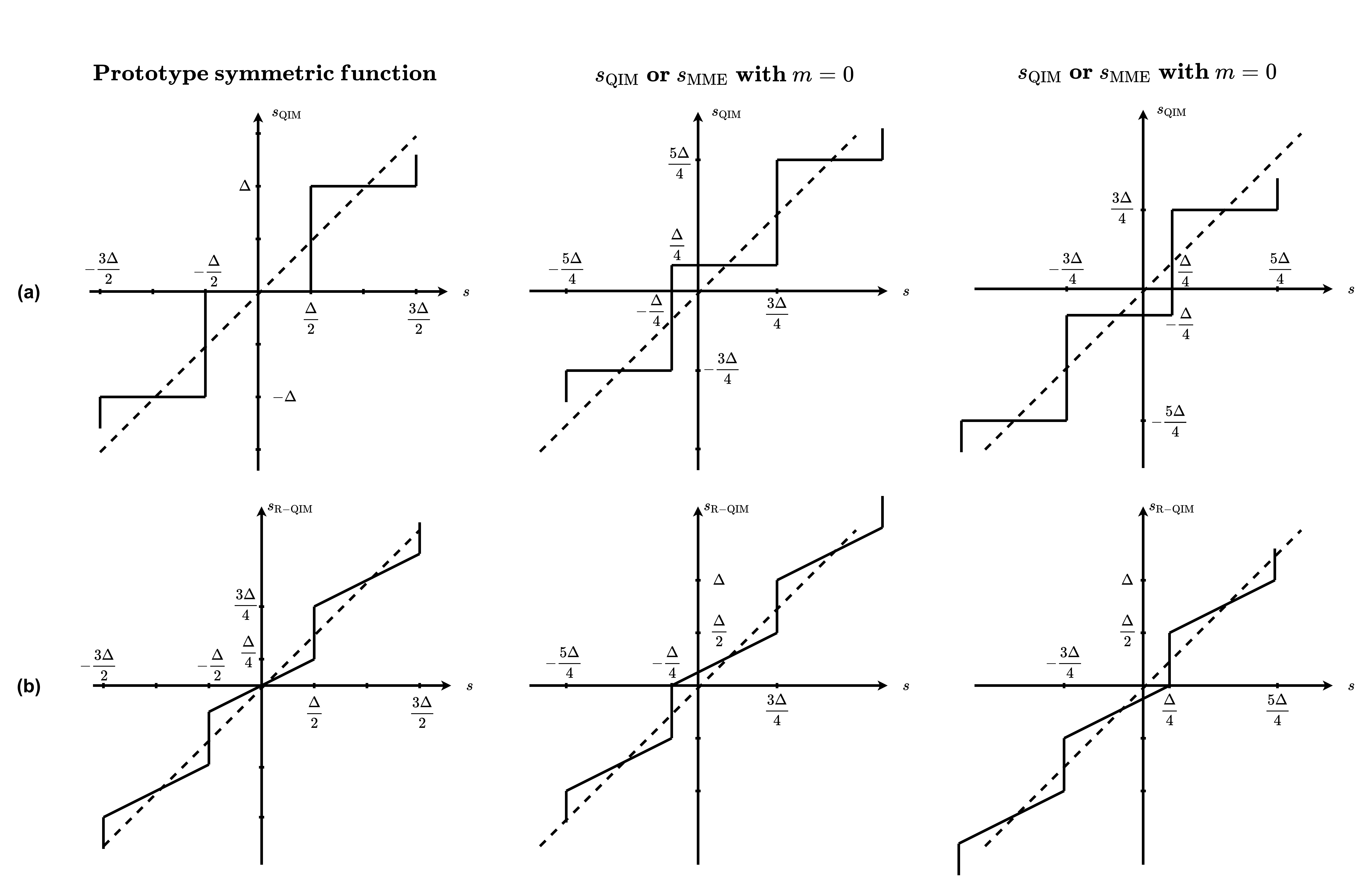} 
	\caption{Selection of watermarked signal with given $s$ and $m \in \{0,1\}$ for Prototype symmetric function, $x=Q_m(s)$ with $m = 0$, and $x=Q_m(s)$ with $m = 1$ in different one-bit watermarking. (a) Conventional QIM. (b) Reversible QIM with $\alpha=0.5$.} 
	\label{Fig. QIM distortion}
\end{figure}
\section{THE PROPOSED METHOD}\label{Sec. Mathmetical Model}

In this section, we introduce a QIM-based RDH (Reversible Data Hiding) algorithm called reversible QIM (R-QIM) and highlight its advantages compared to the method proposed in \cite{Guan_2020}.

\subsection{R-QIM}\label{SSec. QY-RW}

We observe that there exists a quantization error $e$ between the cover vector $s$ and its quantized watermarked vector $Q_{(m,k)}(s)$, given by:

\begin{align}
	e =  s - Q_m(s).
\end{align}

If we only use $Q_m(s)$ as the watermarked vector, the information about $e$ is lost. However, QIM has a certain error tolerance capability. If we consider $e$ as "beneficial noise" and add it back to $Q_m(s)$, we can maintain the information about the cover $s$, making the scheme reversible. The challenge lies in properly scaling the "beneficial noise" to meet specific requirements. First, the scaled $e$ should be small enough to stay within the correct decoding region. Second, the scaled $e$ should not be too small to avoid exceeding the representation accuracy of numbers.

The method that incorporates these ideas is called R-QIM. Its embedding operation is defined as:

\begin{equation}\label{Eq. QY-RW embedding}
	s_{\rm{R-QIM}}\triangleq \alpha Q_{(m,k)}(s)+(1-\alpha)s,
\end{equation}
where $\alpha$ represents a scaling factor that satisfies $\alpha\in \left(\frac{|\mathcal{M}|-1}{|\mathcal{M}|},1\right)$, $Q_{(m,k)}(s)$ is an encrypted quantizer defined as:

\begin{equation}\label{Eq. QY-RW quantizer}
	Q_{(m,k)}(s)\triangleq Q_{\Delta}(s-d_m-k)+d_m+k, m\in \mathcal{M}.
\end{equation}
In Eq. (\ref{Eq. QY-RW quantizer}), $Q_{\Delta}(s)$ denotes the same $Q_{\Delta}(s)=\Delta\lfloor s/\Delta\rfloor$ used in conventional QIM, and $k$ represents a dithering component for secrecy. R-QIM can be considered a fast version of the lattice-based method proposed in \cite{qin2022lattice}.

The parameters $k$ and $\alpha$ are typically treated as secret keys in a watermarking scheme. By setting $k=0$ and $\alpha=0.5$, we can achieve a 1-bit embedding example of R-QIM as depicted in Fig. \ref{Fig. QIM embed model}(b). In this case, the watermarked covers are distributed in the green and red zones around the circle and cross positions, rather than on the positions themselves.

For the receiver, the estimated watermark can be extracted from the received signal $y$ using the following equation:
\begin{equation}\label{Eq. QY-RW extracting}
	\hat{d_m}\equiv Q_{\frac{\Delta}{|\mathcal{M}|}+k}(y)= \left[Q_{\frac{\Delta}{|\mathcal{M}|}}(y-k)+k\right]\{\Delta\}.
\end{equation}

If the noise term $n$ is small enough to satisfy the condition:
\begin{equation}\label{Eq. QY-RW correct decoding}
Q_{\frac{\Delta}{|\mathcal{M}|}+k}(n)=0,
\end{equation}
the correct extraction $\hat{d_m}=d_m$ is achieved, whether in a noiseless or noisy channel.

To estimate the original weight $s$ from the received signal $y$, we use the following equation:
\begin{equation}\label{Eq. QY-RW recovering}
\hat{s}=\frac{y-\alpha Q_{\frac{\Delta}{|\mathcal{M}|}+k}(y)}{1-\alpha}.
\end{equation}
The correct restoration $\hat{s}=s$ occurs if and only if $n=0$ such that $y=s_{\rm{R-QIM}}$. In the presence of noise, the estimation error is given by:

\begin{equation}\label{Eq. QY-RW recovering difference}
\hat{s}-s=\frac{n}{1-\alpha}.
\end{equation}

By setting $\alpha=0.5$ and $k=0$, the embedding distortion is depicted in Fig. \ref{Fig. QIM distortion}(b), with a maximum error of $\alpha\Delta/2=\Delta/4$. If the quantization errors are uniformly distributed over $[-(\Delta/2),(\Delta/2)]$, the mean-squared embedding distortion is:

\begin{equation}\label{Eq. MSE_proposed}
D=\frac{\alpha}{12}\Delta^2.
\end{equation}

Since $\bigcup_{m=0}^{|\mathcal{M}|-1}\Lambda_m=\mathbb{R}$ (as shown in Fig. \ref{Fig. QIM embed model}), each bit of the watermark can be embedded into a host sample with any characteristic and distribution. These features make R-QIM capable of accommodating a watermark of almost the same maximum length as the number of host samples.

\subsection{Discussions}\label{SSec. Theoretical analysis}

In this section, we compare the R-QIM algorithm with the HS algorithm proposed in \cite{Guan_2020} and discuss their respective advantages in terms of usability, capacity, and imperceptibility.

First, let's consider usability. The HS algorithm is not suitable for RDH-based static DNN watermarking due to two main reasons. Firstly, it mismatches the host of uniform distribution, which makes the watermarked sequence exhibit obvious statistical characteristics. This vulnerability makes the algorithm defenseless against passive attacks. Secondly, the low capacity of the HS algorithm becomes even worse when applied to uniformly distributed hosts. Therefore, HS is not feasible for static DNN watermarking, as the data preprocessing operation makes the host sequence uniform rather than normally distributed. We conducted experiments to verify this by preprocessing different randomly generated data of normal distribution, testing them multiple times for skewness, kurtosis, and Kolmogorov-Smirnov (K-S) tests, and plotting the Quantile-Quantile (Q-Q) plot for one of the test results. The results (Fig. \ref{fig_ta_qq_plot}) clearly show that the preprocessed data becomes flatter and deviates from a normal distribution according to the K-S test. Figures \ref{fig_ta_qq_plot}(d), (e), and (f) demonstrate that the preprocessed data follows a uniform distribution. Thus, HS \cite{Guan_2020} lacks practical usability, while R-QIM is feasible for data admitting any distribution.

Next, let's analyze the theoretical advantages of R-QIM in terms of capacity and imperceptibility compared to HS. To evaluate the embedding capacity, we consider a host sequence of length $L$ and analyze the maximum available watermark length $C_{\rm{max}}$ for both R-QIM and HS. In the HS algorithm, the watermark is only embedded into the bin where $s=\Omega_{\mathrm{max}}$, and the other bins do not contain any information about the watermarks.  Recall that Regions i, ii and iii are shown in  
Fig. \ref{fig_HS_shift}(c).
The maximum length of the available watermark in HS can be calculated as:
\begin{equation}
	C_{\rm{max},HS}=\mathrm{Pr}(X \in \mathrm{Region~ii}) \cdot L.
\end{equation}
On the other hand, in R-QIM, the entire host sequence can be used to embed the watermark, resulting in $C_{\rm{max},R-QIM}=L$. It is evident that for host sequences of the same length, R-QIM has a higher embedding capacity.

In terms of embedding distortion or imperceptibility, we define the signal-to-watermark ratio (SWR) as a measure. The SWR is defined as:
\begin{align}\label{Eq. SWR1}
	\mathrm{SWR}~(\mathrm{dB})=10\times \log\left(\frac{\sigma^2_\mathbf{s}}{\sigma^2_\mathbf{w}}\right),
\end{align}
where $\sigma^2_\mathbf{s}$ and $\sigma^2_\mathbf{w}$ represent the power of the host and the additive watermark, respectively. A smaller value of $\sigma^2_\mathbf{w}$ indicates a higher SWR, which implies better imperceptibility. To analyze the embedding distortion fairly, we assume the same capacity and host distribution for both HS in \cite{Guan_2020} and R-QIM, corresponding to embedding the watermark into the host $\Omega_{\mathrm{max}}$ which follows a Gaussian distribution.  

Regarding the embedding distortion, we have the following result:

\begin{theorem}\label{The. greater SWR}
	R-QIM achieves a larger SWR than HS when $\Delta\leq\sqrt{3}$.
\end{theorem}

\begin{proof}	
	Due to the symmetry of the Gaussian distribution, we have $\mathrm{Pr}(X \in \mathrm{Region~ii})=\mathrm{Pr}(X \in \mathrm{Region~iii})$. Therefore, $\mathrm{Pr}(X \in \mathrm{Region~ii})=1-2\mathrm{Pr}(X \in \mathrm{Region~iii})$. With the same settings, the $\sigma^2_\mathbf{w}$ of HS is given by:
	\begin{align}\label{Eq. SWR_HS}
		\sigma^2_{\mathbf{w},\mathrm{HS}}=&\frac{1}{4}\mathrm{Pr}(X \in \mathrm{Region~ii})+\mathrm{Pr}(X \in \mathrm{Region~iii})\notag\\
		=&\frac{1}{4}+\frac{1}{2}\mathrm{Pr}(X \in \mathrm{Region~iii}),
	\end{align}
	while the $\sigma^2_\mathbf{w}$ of R-QIM is given by:
	\begin{align}\label{Eq. SWR_R-QIM}
		\sigma^2_{\mathbf{w},\mathrm{R-QIM}}=&\frac{\alpha\Delta^2}{12}\mathrm{Pr}(X \in \mathrm{Region~ii})\notag\\
		=&\frac{\alpha\Delta^2}{12}-\frac{\alpha\Delta^2}{6}\mathrm{Pr}(X \in \mathrm{Region~iii}).
	\end{align}
	Based on Eqs. (\ref{Eq. SWR_HS}) and (\ref{Eq. SWR_R-QIM}), we have:
	\begin{align}\label{Eq. difference HS2RQIM_2}
		\sigma^2_{\mathbf{w},\mathrm{HS}}-\sigma^2_{\mathbf{w},\mathrm{R-QIM}}
		=\frac{3-\alpha\Delta^2}{12}+\frac{3+\alpha\Delta^2}{6}\mathrm{Pr}(X \in \mathrm{Region~iii}).
	\end{align}
	Since $0<\mathrm{Pr}(X \in \mathrm{Region~iii})<1/2$, we have:
	\begin{eqnarray}
		\frac{3-\alpha\Delta^2}{12}<\frac{3-\alpha\Delta^2}{12}+\frac{3+\alpha\Delta^2}{6}\mathrm{Pr}(X \in \mathrm{Region~iii})<\frac{1}{2}.
	\end{eqnarray}
	Equation (\ref{Eq. difference HS2RQIM_2}) is larger than $0$ when $\Delta\leq\sqrt{3}$. Thus, the theorem is proved.
\end{proof}

According to Theorem \ref{The. greater SWR}, when considering a fixed setting with $\Delta=1$ in HS \cite{Guan_2020}, it can be observed that R-QIM achieves lower embedding distortion and better fidelity, based on the aforementioned assumption. Furthermore, it indicates that the fidelity of R-QIM can be controlled. When
$\Delta>\sqrt{3}$, by adjusting the parameters, we can obtain flexible fidelity performance, whether it is better or worse than HS. In the subsequent scheme design, we will demonstrate the benefits of this feature.

\begin{figure*}[t!]
	\centering
	\subfigure[Skewness]{\includegraphics[width=.3\textwidth]{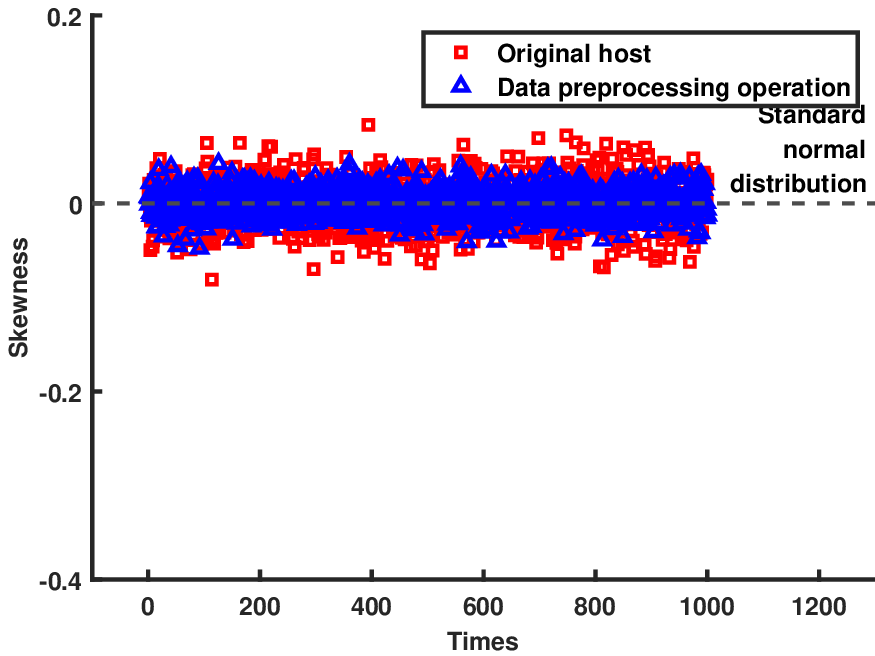}} \quad
	\subfigure[Kurtosis]{\includegraphics[width=.3\textwidth]{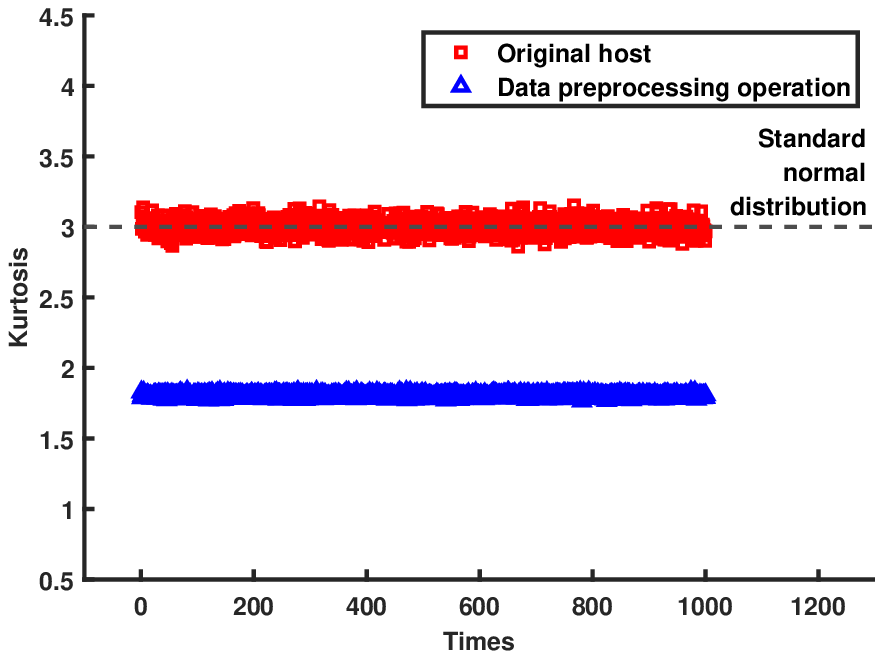}} \quad
	\subfigure[K-S test result] {\includegraphics[width=.3\textwidth]{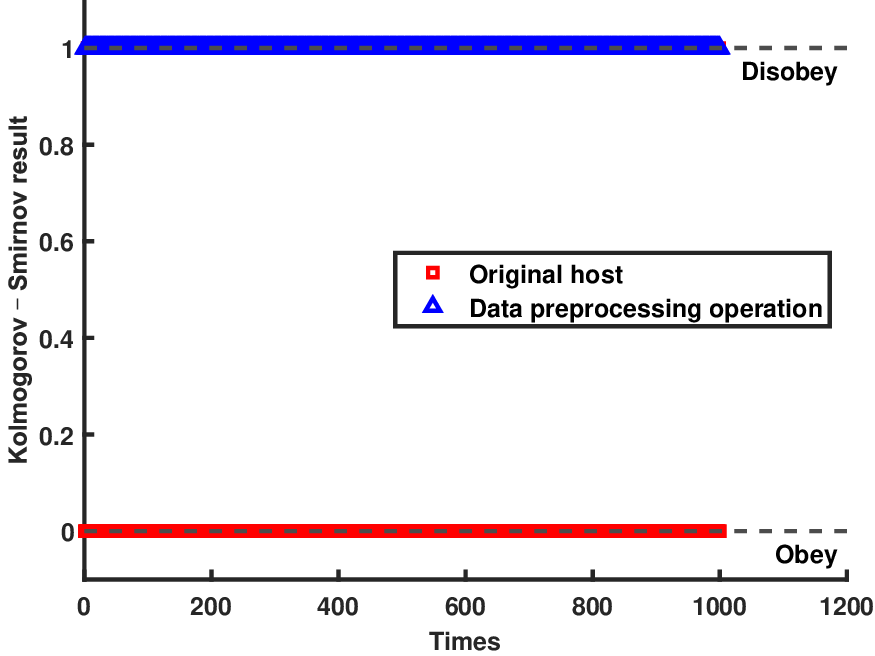}}
	\subfigure[Original versus normal distribution when $c=3$] {\includegraphics[width=.3\textwidth]{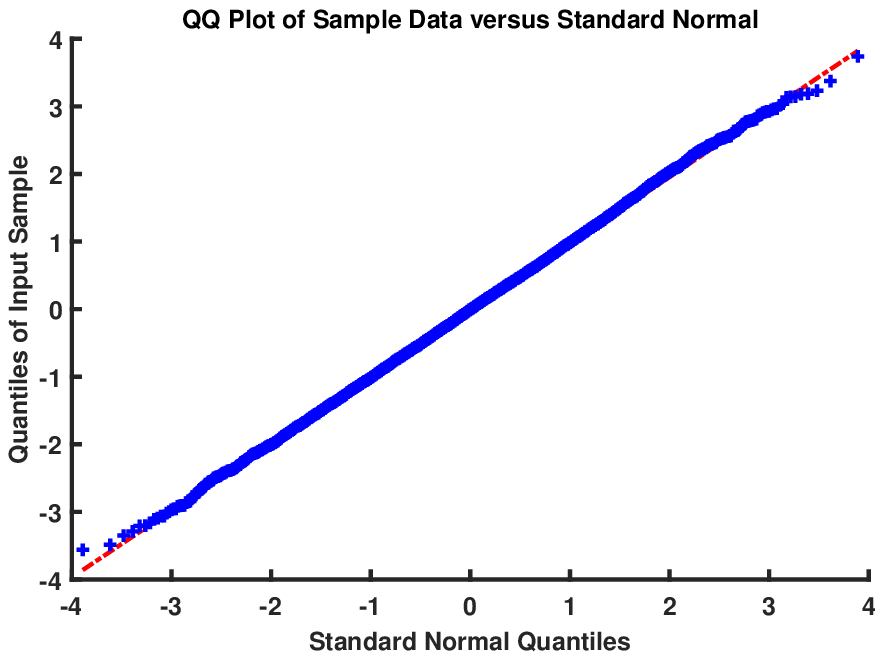}} \quad
	\subfigure[Preprocessed versus normal distribution when $c=3$] {\includegraphics[width=.3\textwidth]{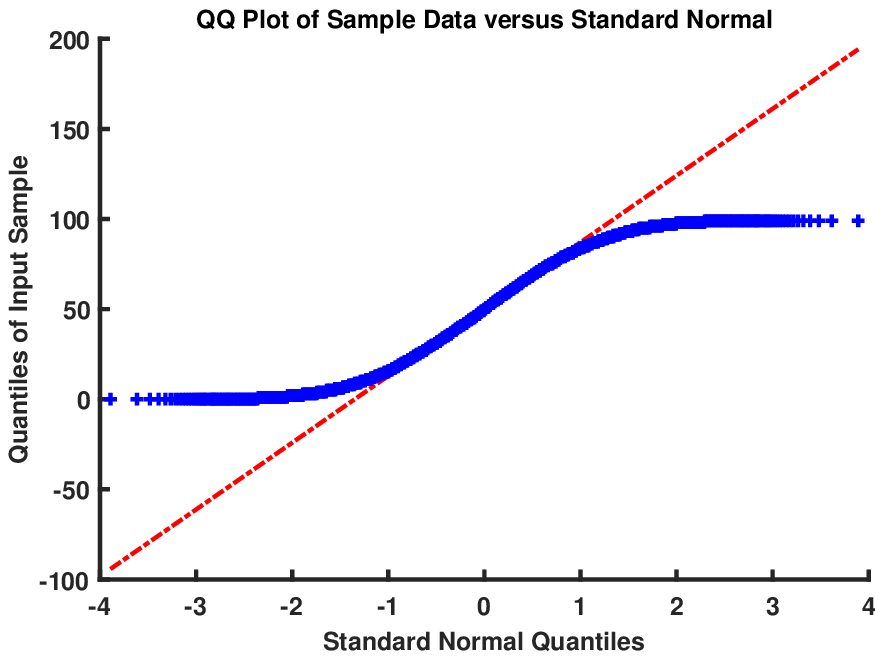}} \quad
	\subfigure[Preprocessed versus uniform distribution when $c=3$] {\includegraphics[width=.3\textwidth]{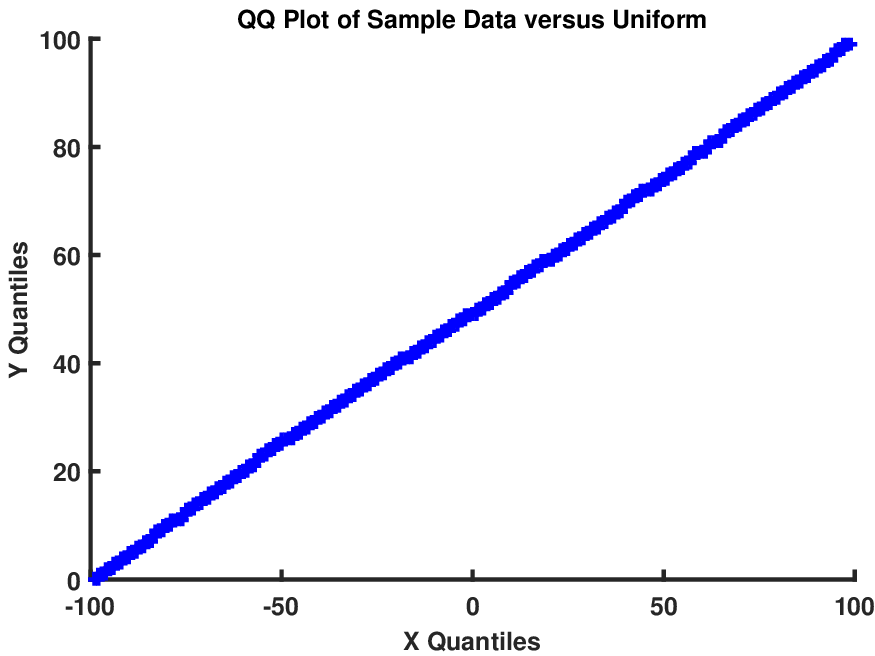}} 
	\caption{Skewness, Kurtosis and K-S test results on  normal distributed data, and their respective  Q-Q plots.} 
	\label{fig_ta_qq_plot}
\end{figure*}

\section{Applications of R-QIM in Static DNN Watermarking}\label{Sec. QY-RW scheme}

In this section, we explore the application of R-QIM in static deep neural network (DNN) watermarking. We propose a scheme that includes several algorithms to facilitate the embedding, extraction, and restoration processes. Furthermore, we outline the concrete steps for functions such as integrity protection and infringement identification, which can be realized using the proposed scheme. The schematics of the two applications are depicted in Figure \ref{Fig. Process}. For the sake of simplicity, we refer to the owner of the DNNs as "Alice," the legal user as "Bob," the illegal user as "Mallory," and the trusted third-party institution as "Institution."

\subsection{Wrapping up R-QIM}

To address security concerns, R-QIM requires additional measures. The watermarking, extracting, and restoring processes based on R-QIM are presented through pseudo-codes in \textbf{Mark} (Algorithm \ref{Alg. Mark}), \textbf{Extract} (Algorithm \ref{Alg. Verify}), and \textbf{Restore} (Algorithm \ref{Alg. Restore}). In these algorithms, certain parameters such as $cl$ and $k$ are set using a pseudo-random number generator (PRNG), while others like the step size $\Delta$ and scaling factor $\alpha$ are determined by the owner.

\textbf{Mark} (Algorithm \ref{Alg. Mark}) takes as input the trained model $\mathbf{W}$, the watermark $\textbf{m}$, and the aforementioned parameters. It outputs the watermarked model $\mathbf{W}_{wtm}$ along with side information, including the watermark information $\mathbf{w\_info}$ and the secret key $\mathbf{sk}$. By selecting a sequence $\mathbf{s}= [s_0,s_1,...,s_{L-1}]$ with the clue $cl$ and extracting relevant information ($L$ and $|\mathcal{M}|$) from $\textbf{m}$ using the $\mathrm{Info}$ function, each bit of the watermark $m_i$ is embedded into $s_i$ using the R-QIM embedding equation (\ref{Eq. QY-RW embedding}) via the $\mathrm{Emb}(\cdot)$ function. The watermark information $\mathbf{w\_info}$ combines $L$ and $|\mathcal{M}|$, while $\mathbf{sk}$ includes $cl$, $k$, and $\Delta$. To maintain the security properties of the embedded watermark, the owner of the DNN model should keep $\mathbf{w\_info}$, $\mathbf{sk}$, and $\alpha$ confidential.

\textbf{Extract} (Algorithm \ref{Alg. Verify}) performs the watermark extraction from the watermarked model $\mathbf{W}_{wtm}$ generated by \textbf{Mark}. With the assistance of the watermark information $\mathbf{w\_info}$ and the secret key $\mathbf{sk}$ held by the owner, an estimated sequence $\hat{\mathbf{d}}$ is created using the R-QIM extraction equation (\ref{Eq. QY-RW extracting}) via the $\mathrm{Ext}(\cdot)$ function, following the same selection process as in \textbf{Mark}. Then, utilizing the watermark information in the codebook, \textbf{Extract} outputs an estimated watermark $\hat{\mathbf{m}}$ derived from $\hat{\mathbf{d}}$. Notably, since watermark extraction requires the assistance of the secret key $\mathbf{sk}$ rather than the scaling

factor $\alpha$, which relates to the security of DNN model recovery, \textbf{Extract} should be performed by the DNN model owner or a trusted third-party institution, ensuring the non-disclosure of the scaling factor $\alpha$.

\textbf{Restore} (Algorithm \ref{Alg. Restore}) takes the watermarked model $\mathbf{W}_{wtm}$ as input and restores it to its original form using the watermark information $\mathbf{w\_info}$, the secret key $\mathbf{sk}$, and the scaling factor $\alpha$ as side information. After the same selection process as \textbf{Mark}, each sample $y_i$ is recovered to $s_i$ one by one using the R-QIM recovery equation (\ref{Eq. QY-RW recovering}) via the $\mathrm{Rec}(\cdot)$ function. Since the correct restoration relies on the noise term $n$, we can detect tampering in the watermarked model under a noiseless channel or a known noisy channel, making the watermarking process reversible for protecting the integrity of the watermarked model. Furthermore, as the restored model no longer contains the watermark, the effectiveness of the restoration process can be evaluated by verifying the absence of the watermark in the DNN model.

\begin{algorithm}[t]
	\caption{Mark} 
	\label{Alg. Mark}
	\hspace*{0.02in} {\bf Input:} 
	Trained Model $\mathbf{W}$, Watermark $\mathbf{m}$, Scaling Factor $\alpha$, Dithering Vector $k$, Embedding Clue $cl$, Step Size $\Delta$\\
	\hspace*{0.02in} {\bf Output:} 
	Watermarked Model $\mathbf{W}_{wtm}$, Watermark Information $\mathbf{w\_info}$, Secret Key $\mathbf{sk}$
	\begin{algorithmic}[1]
		\State $\left[~L,|\mathcal{M}|~\right]\leftarrow \mathrm{Info}(\mathbf{m})$ 
		\State $\mathbf{c}\leftarrow\mathrm{Construction}(cl,L)$
		\State $\mathbf{W}_{wtm}\leftarrow\mathbf{W}$,  $i\leftarrow0$
		\For{$++i\leq L$} 
		\State $s_i\leftarrow\mathbf{W}(\mathbf{c}(i))$
		\State $m_i\leftarrow\mathbf{m}(i)$
		\State $\mathbf{W}_{wtm}(\mathbf{c}(i))\leftarrow\mathrm{Emb}(s_i,m_i,\alpha,k,\Delta)$
		\EndFor
		\State $\mathbf{w\_info}\leftarrow\left[~L,|\mathcal{M}|~\right]$
		\State $\mathbf{sk}\leftarrow\left[~k,cl,\Delta~\right]$
	\end{algorithmic}
\end{algorithm}

\begin{algorithm}[t]
	\caption{Extract} 
	\label{Alg. Verify}
	\hspace*{0.02in} {\bf Input:} 
	Watermarked Model $\mathbf{W}_{wtm}$, Watermark Information $\mathbf{w\_info}$, Secret Key $\mathbf{sk}$\\
	\hspace*{0.02in} {\bf Output:} 
	Extracted Watermark $\hat{\mathbf{m}}$
	\begin{algorithmic}[1]
		\State $\left[~L,|\mathcal{M}|~\right]\leftarrow \mathbf{w\_info}$ 
		\State $\left[~\alpha,k,cl,\Delta~\right]\leftarrow\mathbf{sk}$
		\State $\mathbf{c}\leftarrow\mathrm{Construction}(cl,L)$
		\State $i\leftarrow0$
		\For{$++i\leq L$} 
		\State $y_i\leftarrow\mathbf{W}_{wtm}(\mathbf{c}(i))$
		\State $d_i\leftarrow\mathrm{Ext}(y_i,k,\Delta,|\mathcal{M}|)$
		\State $\hat{\mathbf{m}}(i)\leftarrow\mathrm{Codebook}(d_i)$
		\EndFor
	\end{algorithmic}
\end{algorithm}

\begin{algorithm}[t]
	\caption{Restore} 
	\label{Alg. Restore}
	\hspace*{0.02in} {\bf Input:} 
	Watermarked Model $\mathbf{W}_{wtm}$, Watermark Information $\mathbf{w\_info}$, Secret Key $\mathbf{sk}$, Scaling Factor $\alpha$\\
	\hspace*{0.02in} {\bf Output:} 
	Recovered Model $\hat{\mathbf{W}}$
	\begin{algorithmic}[1]
		\State $\left[~L,|\mathcal{M}|~\right]\leftarrow \mathbf{w\_info}$ 
		\State $\left[~\alpha,k,cl,\Delta~\right]\leftarrow\mathbf{sk}$
		\State $\mathbf{c}\leftarrow\mathrm{Construction}(cl,L)$
		\State $\hat{\mathbf{W}}\leftarrow\mathbf{W}_{wtm}$, $i\leftarrow0$
		\For{$++i\leq L$} 
		\State $y_i\leftarrow\mathbf{W}_{wtm}(\mathbf{c}(i))$
		\State $\hat{\mathbf{W}}(\mathbf{c}(i))\leftarrow\mathrm{Rec}(y_i,k,\alpha,\Delta,|\mathcal{M}|)$
		\EndFor
	\end{algorithmic}
\end{algorithm}

\begin{figure*}[t!]
	\centering
	\subfigure[Integrity protection]{\includegraphics[width=.9\textwidth]{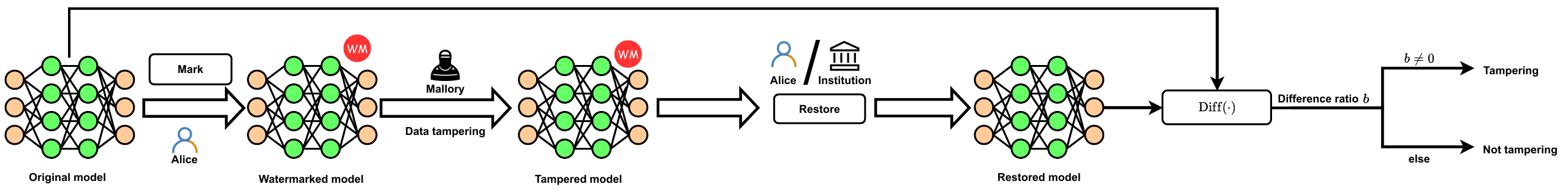}} \quad
	\subfigure[Infringement identification]{\includegraphics[width=.9\textwidth]{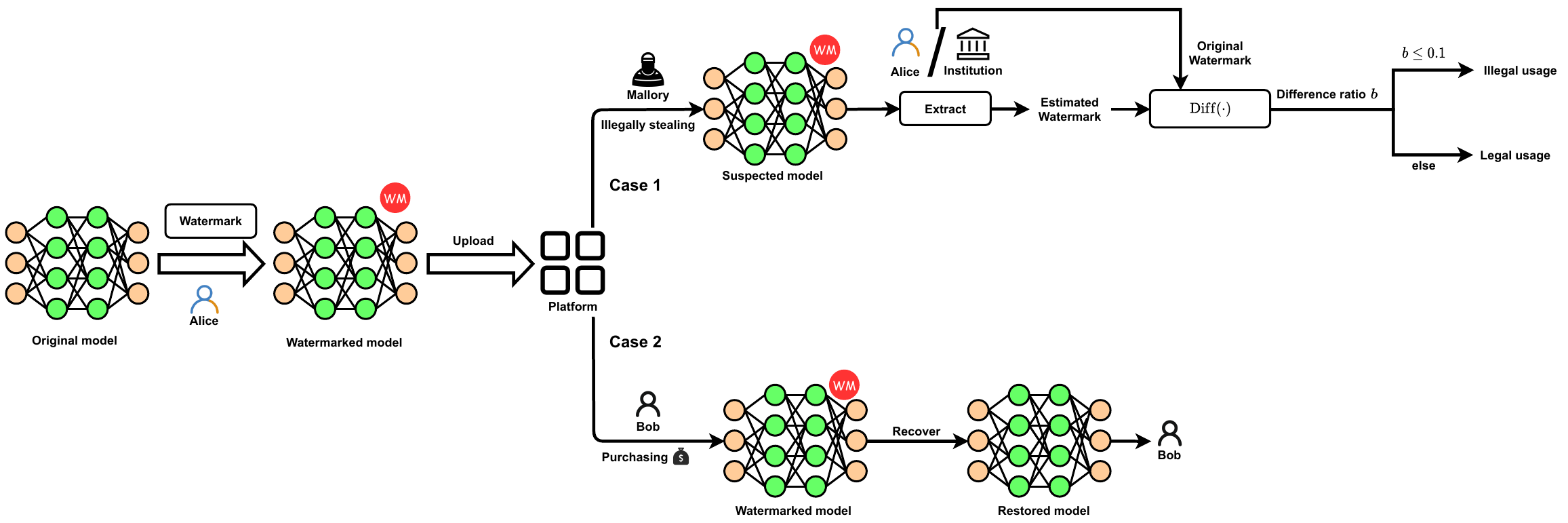}}\quad 
	\caption{Schematics of using R-QIM for integrity protection and infringement identification.} 
	\label{Fig. Process}
\end{figure*}
 
\subsection{Integrity Protection}

To enable reversibility in DNN watermarks, Guan et al. \cite{Guan_2020} introduced the concept of integrity protection for DNN models. They proposed a scheme that verifies whether a DNN model has been tampered with by comparing the bit differences in weights between the restored and original models. Building on this idea, we present an integrity protection scheme that employs R-QIM \cite{Guan_2020}, as depicted in Figure \ref{Fig. Process} (a).

In this scheme, Alice embeds her watermarks into a commercialized DNN model $\mathbf{W}$ using the \textbf{Mark} algorithm \cite{Guan_2020}, resulting in a watermarked DNN model $\mathbf{W}_{wtm}$. During transmission, Mallory illegally intercepts $\mathbf{W}_{wtm}$, modifies its weights, and profits from sharing the tampered model. To identify tampering, we define two types of operations for noiseless and noisy channels.

\begin{itemize}
	\item \textbf{Noiseless channel:} In this scenario, where correct recovery is guaranteed, the tampered DNN model is restored using the \textbf{Restore} algorithm \cite{Guan_2020} to obtain an estimated model. Notably, the \textbf{Restore} function can meet the requirements of perfect recovery after watermark extraction since Equation (\ref{Eq. QY-RW recovering}) \cite{Guan_2020} contains $Q_{\frac{\Delta}{|\mathcal{M}|}+k}(y)$, which can be regarded as watermark extraction. The weights of the restored model are then compared to the original model using a difference function $\mathrm{Diff}(\cdot)$, which calculates a difference ratio $b$. Due to the sensitivity of the recovery process, even minor changes to $\mathbf{W}_{wtm}$ would lead to differences in the weights of the restored model. This characteristic allows for integrity assessment, where $b=1$ (or $b=0$) indicates that $\mathbf{W}_{wtm}$ has (not) been tampered with.
	\item \textbf{Noisy channel:} In this scenario, tampering of DNN models can be identified when the noise term $n$ is sufficiently small. By leveraging Equation (\ref{Eq. QY-RW recovering difference}) \cite{Guan_2020}, the difference between the restored and original models can be theoretically measured, allowing for a comparison that excludes the interference of the noise term $n$. Theoretical differences between the restored and original models are computed using Equation (\ref{Eq. QY-RW recovering difference}) \cite{Guan_2020}, and a difference ratio $b$ is obtained. When $b=1$ (or $b=0$), it indicates that $\mathbf{W}_{wtm}$ has (not) been tampered with.
\end{itemize}

In summary, our proposed scheme is well-suited for integrity protection compared to the scheme presented by Guan et al. \cite{Guan_2020}. Our scheme offers higher fidelity for $\mathbf{W}_{wtm}$, as justified by Theorem \ref{The. greater SWR} \cite{Guan_2020}. Additionally, our scheme is the first to protect the integrity of $\mathbf{W}_{wtm}$ over noisy channels.

\subsection{Infringement Identification}
In addition to integrity protection, reversible DNN watermarking can be utilized for infringement identification of suspicious DNN models. When the watermark is removed during the recovery operation, no watermark remains in the restored model. This enables the distinction between a legal user holding the restored model and an illegal user holding the watermarked model. Based on this concept, we propose a novel scheme for infringement identification of DNN models, where a user receives a secret key for recovery after legalization and obtains a restored model. The proposed scheme for legitimate authentication is illustrated in Figure \ref{Fig. Process} (b). For simplicity, we refer to the owner, legal user, illegal user, and trusted third-party institution as "Alice," "Bob," "Mallory," and "Institution," respectively.

In our proposed scheme, Alice sells her commercialized DNN model $\mathbf{W}$ through an online/offline platform and utilizes our scheme for marking the ownership of $\mathbf{W}$. After embedding a watermark into $\mathbf{W}$ using the \textbf{Mark} algorithm \cite{Guan_2020}, the resulting watermarked model $\mathbf{W}_{wtm}$ is sent to the platform, serving as an exhibit or trial product to promote Alice's model. As the embedding process occurs after model training, the fidelity of $\mathbf{W}_{wtm}$ is intentionally lower, ensuring that its disclosure does not harm Alice's rights.

When Bob expresses interest in the product, Alice shares $\mathbf{w\_info}$, $\mathbf{sk}$, and $\alpha$ with him. Bob can then recover $\mathbf{W}_{wtm}$ to its original form using the \textbf{Restore} algorithm \cite{Guan_2020}. The recovered model is identical to the original model, maximizing its effectiveness, and the watermark is completely removed, making it undetectable in the recovered model. If Mallory illegally steals the DNN model and shares it on public platforms, Alice can report the incident to the Institution for arbitration. To authenticate the legitimacy of the suspicious model held by Mallory, Alice or the Institution can extract the estimated watermark $\hat{\textbf{m}}$ from the suspect model using the \textbf{Extract} algorithm \cite{Guan_2020}. Then, $\hat{\textbf{m}}$ can be compared to Alice's watermark using $\mathrm{Diff}(\cdot)$, which outputs a difference ratio $b$ for detecting the presence of the embedded watermark. When $b\leq0.1$, the watermark is considered detected, and Mallory is identified as an illegal user.

To avoid infringing on Alice's rights, the DNN watermarking scheme for infringement identification should exhibit lower fidelity, ensuring that the effectiveness of the watermarked model is no better than the original one. Thanks to Theorem \ref{The. greater SWR} \cite{Guan_2020}, R-QIM offers greater distortion than HS \cite{Guan_2020} by setting $\Delta>\sqrt{3}$ and an appropriate $\alpha$, making it more suitable for the infringement identification scenario.

\section{SIMULATIONS}\label{Sec. experiment}

\begin{figure*}[t!]
	\centering
	\subfigure[Original MLP weigts versus normal distribution]{\includegraphics[width=.3\textwidth]{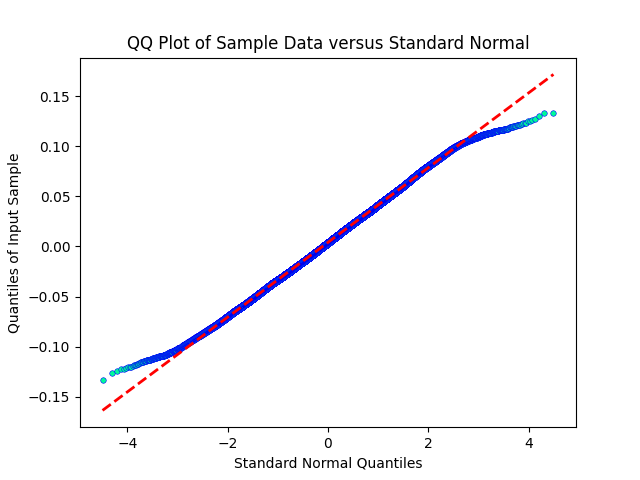}} \quad
	\subfigure[Preprocessed MLP weights versus normal distribution]{\includegraphics[width=.3\textwidth]{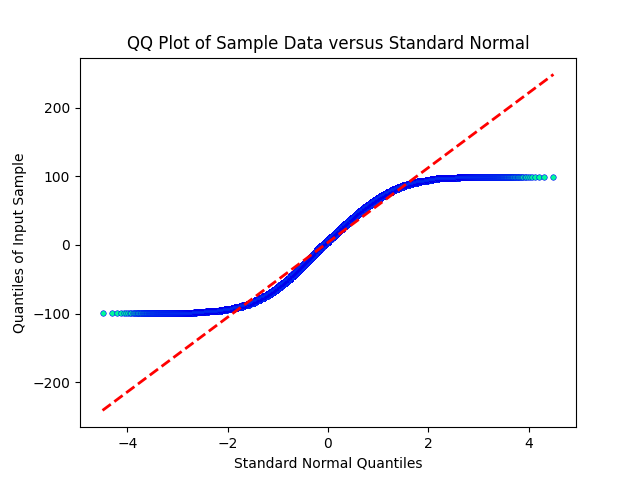}}\quad 
	\subfigure[Preprocessed MLP weights versus uniform distribution]{\includegraphics[width=.3\textwidth]{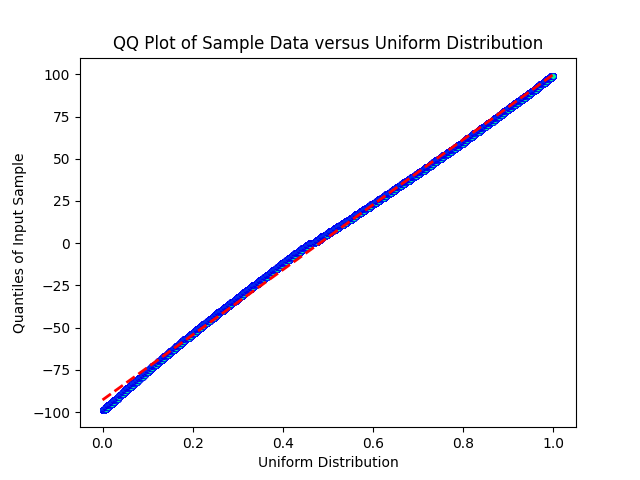}} 
	\subfigure[Original VGG weigts versus normal distribution]{\includegraphics[width=.3\textwidth]{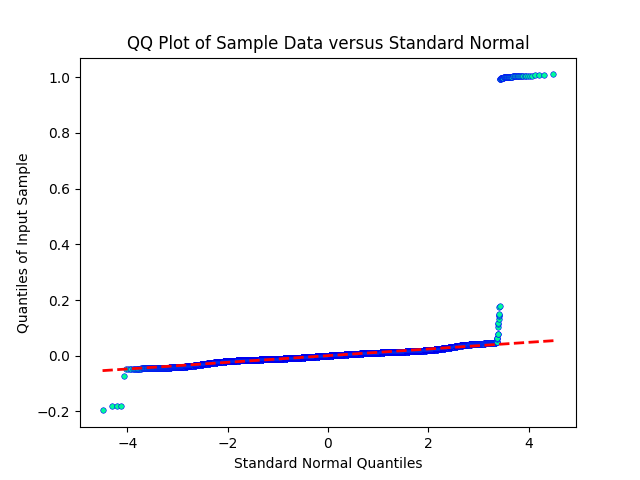}} \quad
	\subfigure[Preprocessed VGG weights versus normal distribution]{\includegraphics[width=.3\textwidth]{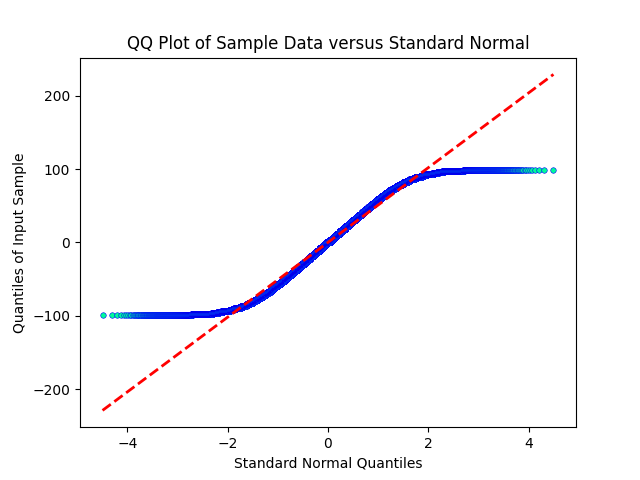}} \quad
	\subfigure[Preprocessed VGG weights versus uniform distribution]{\includegraphics[width=.3\textwidth]{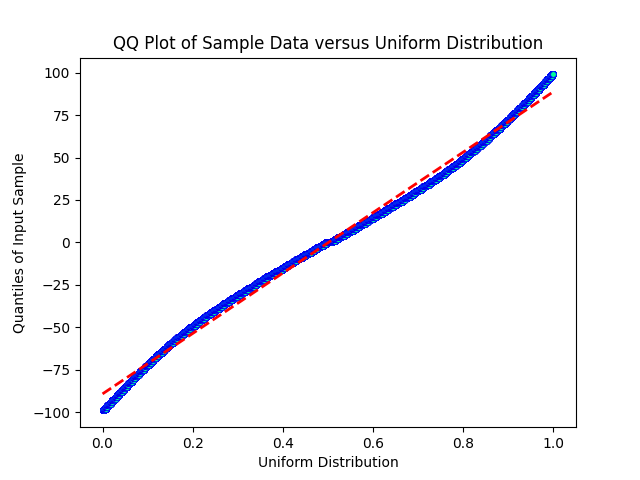}} 
	\caption{Q-Q plot on the original and preprocessed weights with $c=3$.} 
	\label{fig_ss1_qq_plot}
\end{figure*}

To evaluate the effectiveness of the proposed R-QIM scheme, the simulations are divided into three parts:
i) Usability of the HS method in \cite{Guan_2020}.
ii) Comparison between R-QIM and HS in terms of capacity and fidelity.
iii) Impact of R-QIM parameters.

The experimental setups for these simulations are summarized as follows:

\textbf{Datasets and Models}: The datasets chosen for training the models are \textbf{MNIST} \cite{mnist} and \textbf{CIFAR10} \cite{krizhevsky2009learning}. The \textbf{MNIST} dataset consists of 60,000 training and 10,000 testing gray-scale images of handwritten digits, each with a size of 28$\times$28 pixels and divided into 10 classes. The \textbf{CIFAR10} dataset contains 50,000 training and 10,000 testing color images of various objects, with a size of 32$\times$32 pixels. Two combinations of models and datasets are used: the Multi-layer Perceptron (MLP) model trained on \textbf{MNIST} (referred to as Group A) and the Visual Geometry Group (VGG) model trained on \textbf{CIFAR10} (referred to as Group B).

\textbf{Parameters}: The watermark is generated by converting a piece of text into a bit stream. The step size is set to $\Delta=1$, and the scaling factor is chosen as $\alpha=0.8675$ with a dithering value of $k=0$.

\textbf{Indicators}: The fidelity of the watermarked model is evaluated using training loss and classification accuracy. The training loss measures the damage caused by watermark embedding, while the classification accuracy reflects the effectiveness of the watermarked network. Lower training loss indicates less impact from watermark embedding, and higher classification accuracy indicates greater effectiveness. To detect the existence of copyright information and perform tampering detection, the bit error rate (BER) is employed. BER is calculated as the ratio of the number of differing bits between the original and estimated message to the total number of bits. In copyright protection, if BER is not larger than 10\%, it indicates the presence of information; otherwise, it does not. For tampering detection, a model is considered untampered if BER equals 0.

\subsection{Usability test}

\begin{table}[t!]\scriptsize
	\setlength{\tabcolsep}{1.5mm}
	\centering
	\renewcommand\arraystretch{1.5}
	\caption{The analysis of weights in different models.}  
	\label{Tab. analysis of weights}  
	\begin{tabular}{c c c c c}
		\hline
		\multirow{2}{*}{\textbf{Metric}} & \multicolumn{2}{c}{\textbf{MLP ($198656$ length)}} & \multicolumn{2}{c}{\textbf{VGG ($200000$ length)}}\\ \cline{2-5}
		&Original &Preprocessed&Original &Preprocessed\\\hline
		Skewness &0.0426 &-0.1059 &34.8544 &-0.0028 \\
		Kurtosis &2.8642 &1.8324 &1692.39 &1.8276 \\
		$P\leq5\%$ in K-S test 
		&\scalebox{0.85}[1]{$\times$} &\scalebox{0.85}[1]{$\times$} &\scalebox{0.85}[1]{$\times$} &\scalebox{0.85}[1]{$\times$} \\
		$P\leq5\%$ in J-B test 
		&\scalebox{0.85}[1]{$\times$} &\scalebox{0.85}[1]{$\times$} &\scalebox{0.85}[1]{$\times$} &\scalebox{0.85}[1]{$\times$} \\
		\hline
	\end{tabular}
\end{table}

\begin{table}[t!]\scriptsize
	\setlength{\tabcolsep}{1.5mm}
	\centering
	\renewcommand\arraystretch{1.5}
	\caption{Comparison of information embedding capacity.}  
	\label{Tab. fig_ss2_capacity_RQIM2HS}  
	\begin{tabular}{c c c c c}
		\hline
		\multirow{2}{*}{\makecell[c]{\textbf{Host}\\\textbf{length}}} & \multicolumn{2}{c}{\textbf{MLP ($198656$ length)}} & \multicolumn{2}{c}{\textbf{VGG ($200000$ length)}}\\ \cline{2-5}
		&R-QIM &HS\cite{Guan_2020}&R-QIM &HS\cite{Guan_2020}\\\hline		
		$20$\% &39732 &377 &40000 &396 \\
		$40$\% &79463 &811 &80000 &760 \\
		$60$\% &119194 &1187 &120000 &1177 \\
		$80$\% &158925 &1565 &160000 &1604 \\
		$100$\% &198656 &1969 &200000 &1982 \\ \hline
	\end{tabular}
\end{table}

 In Section \ref{SSec. Theoretical analysis}, we theoretically identified potential weaknesses of the HS method. To support our claims, we conducted a numerical analysis of the weights of DNN models.
 
 Figure \ref{fig_ss1_qq_plot} presents the Q-Q plot comparing the original and preprocessed weights of the MLP and VGG models against normal and uniform distributions. We observe that the original weights of both models exhibit a distribution closer to the normal distribution, while the preprocessed weights clearly follow the uniform distribution as integers.
 
 Furthermore, we performed an analysis of skewness, kurtosis, K-S test, and Jarque-Bera (J-B) test results for the original and preprocessed weights of the MLP and VGG models. The summarized analysis is presented in Table \ref{Tab. analysis of weights}. The results indicate that: 
 
 i) Data preprocessing flattens the distribution of the weights, resulting in lower kurtosis values for the preprocessed weights.
 
 ii) Neither the original nor the preprocessed weights of the MLP and VGG models pass the K-S and J-B tests.
 
 These experiments demonstrate that realistic DNN weights do not follow a normal distribution. However, this does not undermine the validity of the data preprocessing method proposed in \cite{Guan_2020} for transforming data from a normal to a uniform distribution. Nevertheless, it highlights the limited usability of the \cite{Guan_2020} method.

\subsection{Capacity and fidelity comparisons}

\begin{figure}[t!]
	\centering
	\subfigure[Loss in Group A]{\includegraphics[width=.24\textwidth]{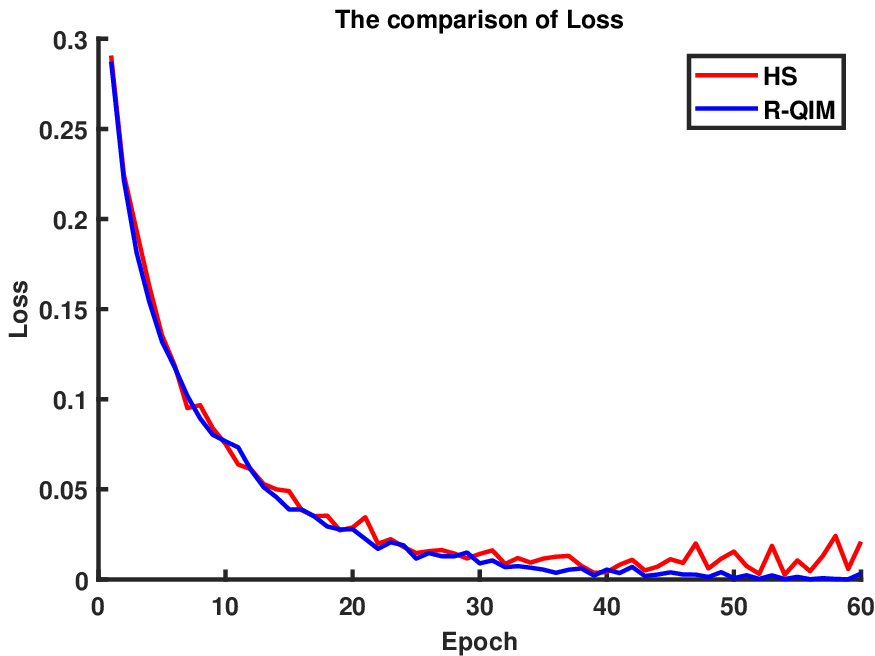}} 
	\subfigure[Accuracy in Group A] {\includegraphics[width=.24\textwidth]{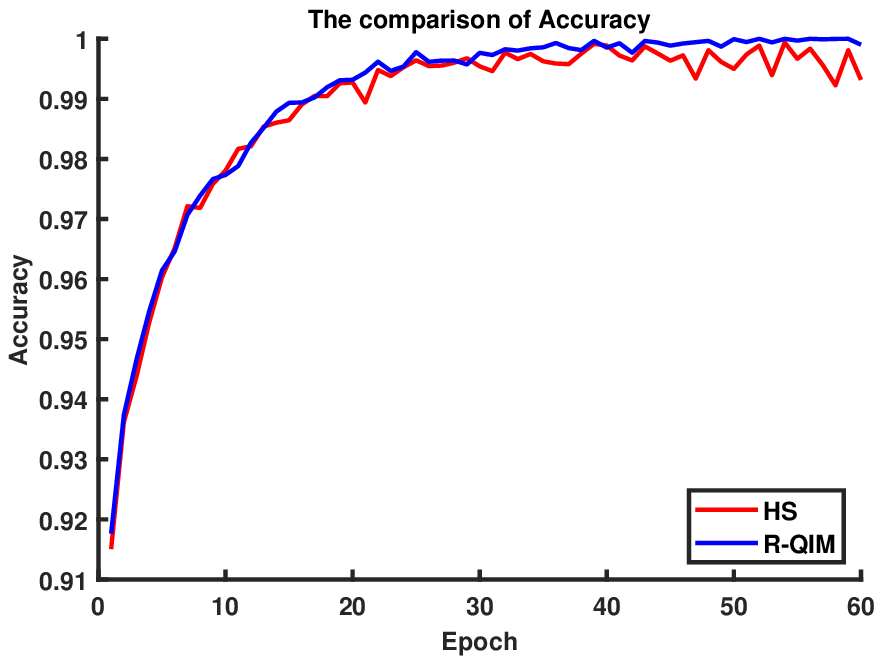}}
	\subfigure[Loss in Group B]{\includegraphics[width=.24\textwidth]{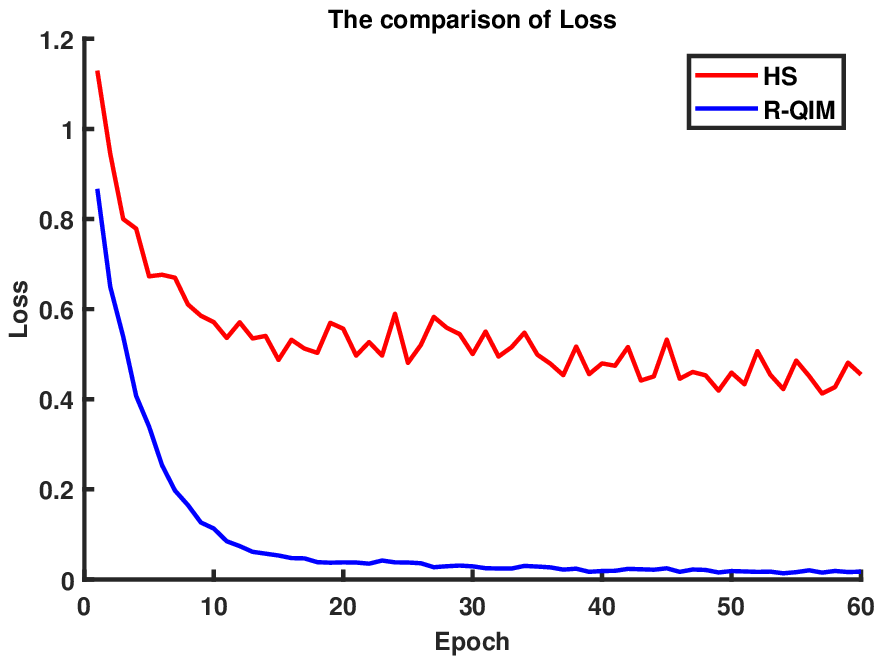}} 
	\subfigure[Accuracy in Group  B] {\includegraphics[width=.24\textwidth]{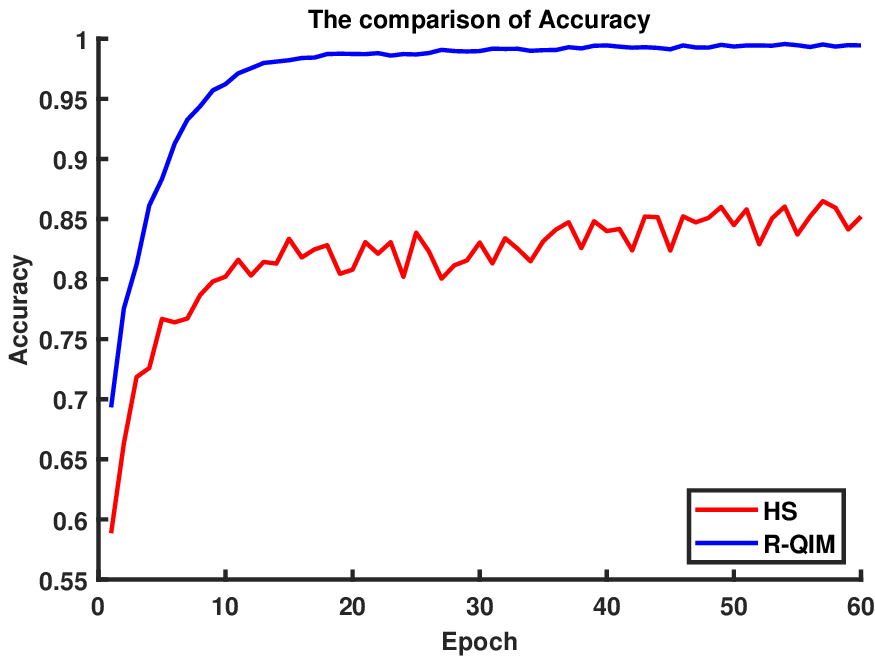}}

	\caption{Comparison of loss and accuracy with different epochs.} 
	\label{fig_ss2_comparison_RQIM2HS1}
\end{figure}

\begin{figure}[t!]
	\centering
	\subfigure[Loss in Group A]{\includegraphics[width=.24\textwidth]{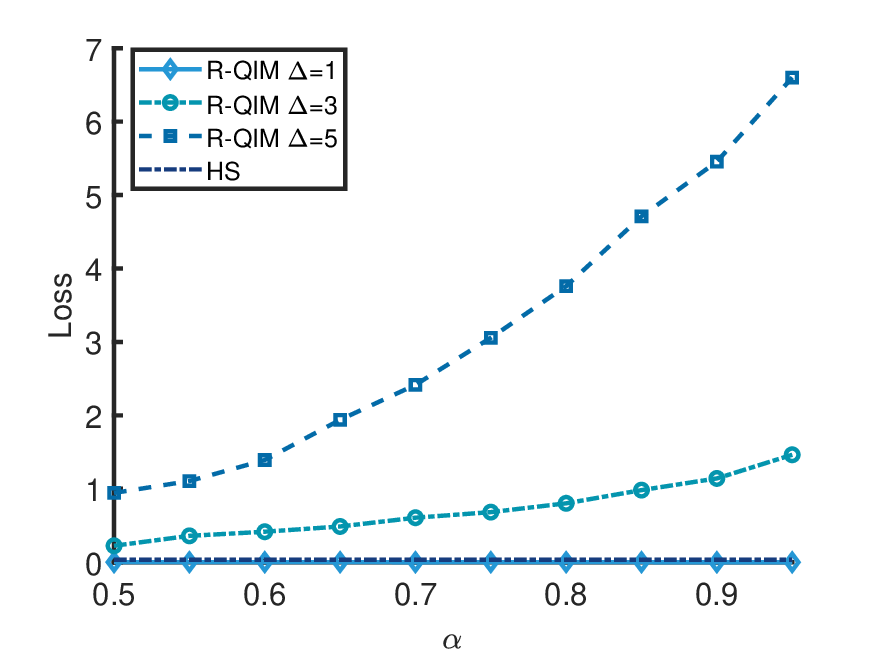}} 
	\subfigure[Accuracy in Group A] {\includegraphics[width=.24\textwidth]{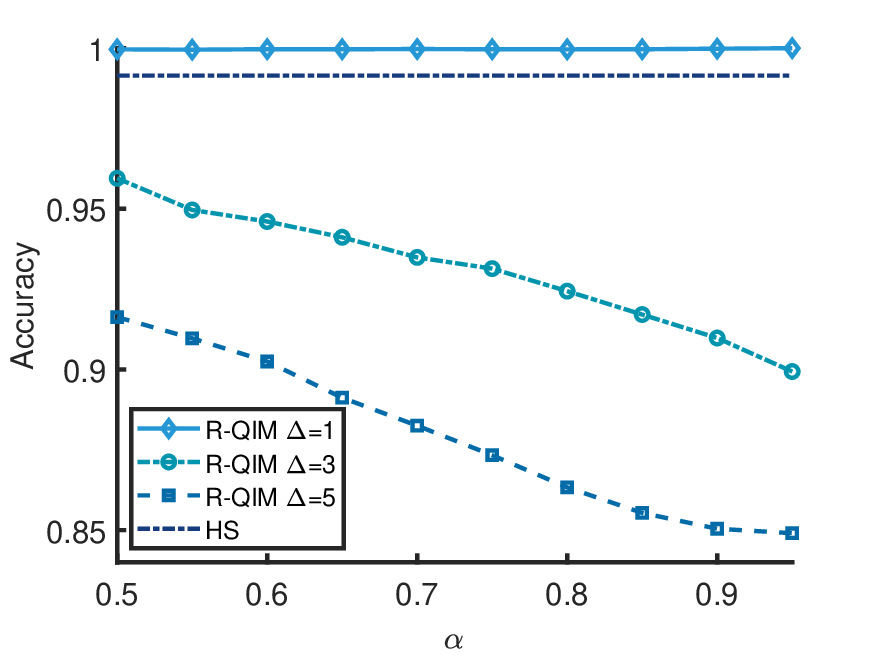}}
	\subfigure[Loss in Group B]{\includegraphics[width=.24\textwidth]{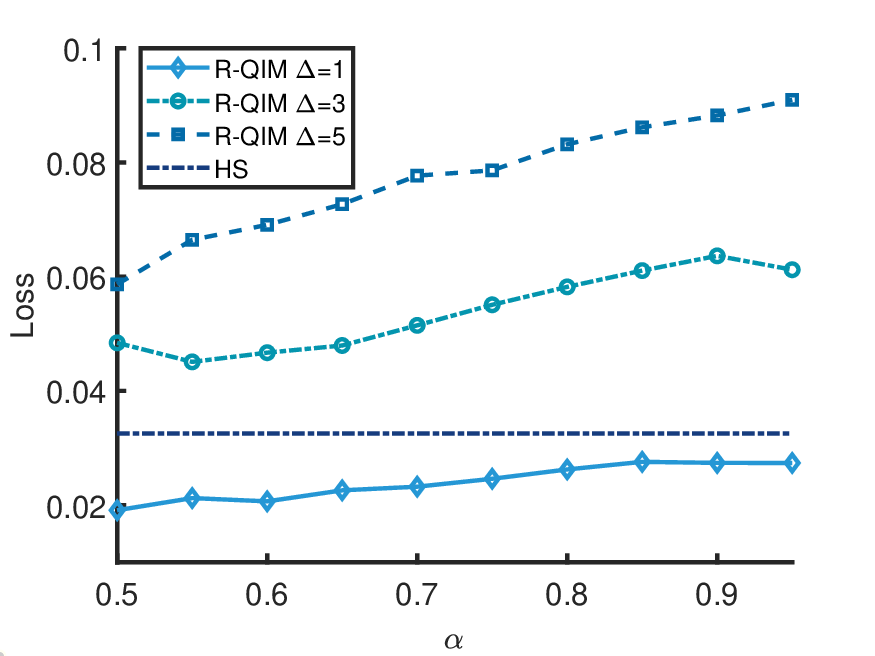}} 
	\subfigure[Accuracy in Group  B] {\includegraphics[width=.24\textwidth]{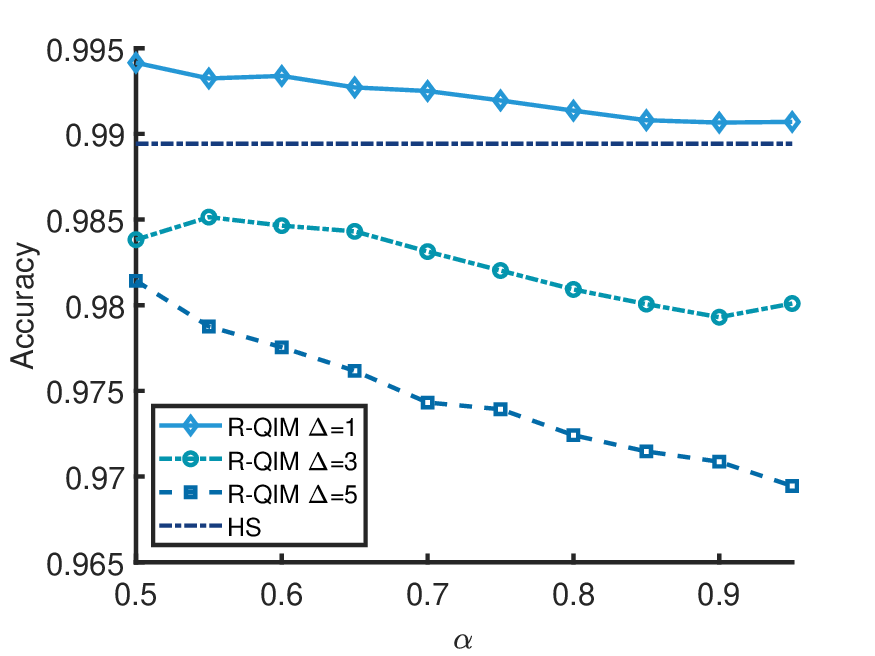}}
	\caption{Comparison of loss and accuracy with different $\alpha$ and $\Delta$.} 
	\label{fig_ss2_comparison_RQIM2HS2}
\end{figure}

In order to assess the adaptability and superiority of the two proposed schemes in terms of capacity and fidelity, we conducted several experiments to compare R-QIM with HS in two specific applications. The first application focuses on integrity protection, which requires a watermarking method with a high embedding capacity and minimal embedding damage. Therefore, we compared the maximum available capacity, classification accuracy, and training loss between our proposed scheme and the method proposed by \cite{Guan_2020} with a fixed step size $\Delta=1$.

The results for the maximum available capacity are presented in Table \ref{Tab. fig_ss2_capacity_RQIM2HS}, revealing a significant difference between R-QIM and HS. Regardless of whether it is group A or B, the table clearly indicates that R-QIM exhibits a higher embedding capability compared to the benchmark method, which aligns with our theoretical analysis.

Regarding fidelity, we compared the training loss and classification accuracy of the watermarked models embedded at different epochs using R-QIM and HS. The corresponding results are depicted in Fig. \ref{fig_ss2_comparison_RQIM2HS1}. It can be observed that R-QIM consistently outperforms the benchmark method in both metrics for both group A and B, with the difference being more pronounced in group B. Importantly, these observations support our assumption that lower embedding distortion leads to better fidelity of the watermarked model.

For infringement identification, according to Theorem \ref{The. greater SWR}, R-QIM can achieve a more noticeable decline in fidelity compared to HS by setting $\Delta>\sqrt{3}$. To verify this, we examined the classification accuracy and training loss of MLP and VGG models with different values of $\alpha$ and $\Delta$, as shown in Fig. \ref{fig_ss2_comparison_RQIM2HS2}. Based on these observations, we conclude the following:
i) As $\alpha$ and $\Delta$ increase, the loss of the watermarked model increases while the accuracy decreases, which aligns with our expectation regarding the relationship between distortion and fidelity.
ii) When $\Delta=1<\sqrt{3}$ in R-QIM, it outperforms HS in terms of both loss and accuracy for groups A and B. However, when $\Delta=3$ and $5>\sqrt{3}$, HS performs better than R-QIM. This finding supports Theorem \ref{The. greater SWR} and demonstrates the flexible performance of R-QIM, which determines its applicability in infringement identification.

\subsection{Performance of R-QIM Recovery}

To assess the performance of R-QIM recovery, we conducted several simulations focusing on the accuracy of the recovered values and the presence of watermarks. In these experiments, the watermarks were converted to uniform data consisting of 4264 bits.

In the first experiment, we aimed to compare the performance difference in implementing reversible operations. We trained two combinations from scratch twice for 60 epochs and embedded the watermark at epoch 30 using the proposed scheme. In the first run (denoted by the green line), a reversible operation was applied immediately after the embedding process, while in the second run (denoted by the red line), no reversible operation was applied. Fig. \ref{Fig. Recovering Effects1} presents notable observations from this experiment:
i) After watermark embedding, the model's accuracy sharply degraded due to the parameter modification. However, with the implementation of the proposed reversible operation, the reduced accuracy immediately restored. This demonstrates the effectiveness of the reversible operation in offsetting the damages caused by watermark embedding.
ii) Without the reversible operation, when the reduced accuracy reached a plateau, the accuracy of both groups dropped below that of the reversible operation applied case. This indicates incomplete compensation in subsequent training for the watermarked model without reversible operation, while the compensation is complete with the reversible operation.
iii) We observed that group A was more severely affected than group B after watermark embedding, and it reached a slower plateau in subsequent training, suggesting a higher effectiveness of the reversible operation in group A.

To analyze the specific effects of the various processes in the proposed method, we compared the values of the original, watermarked, and recovered weights for the two combinations in Fig. \ref{Fig. Recovering Effects2}. In this figure, the points representing the original and recovered cover (represented by a horizontal line and a vertical bar, respectively) coincide at each index of the sample, indicating correct recovery. Additionally, the distribution of watermarked weights (represented by crosses) in Fig. \ref{Fig. Recovering Effects2} illustrates the impact on the weights caused by watermark embedding.

Finally, as the infringement identification function relies on determining whether the restored DNN model contains a watermark, we compared the bit error rate (BER) metric of the watermark with and without the reversible operation, as depicted in Fig. \ref{Fig. Recovering Effects3}. The BER value of the watermarked model was 0.0005, whereas it increased to 0.43 after applying the reversible operation. This confirms that the reversible operation can effectively remove the watermark embedded in the host model, thereby demonstrating the validity of the legitimacy authentication scheme.

\begin{figure}[t!]
	\centering
	\subfigure[Group A]{\includegraphics[width=.24\textwidth]{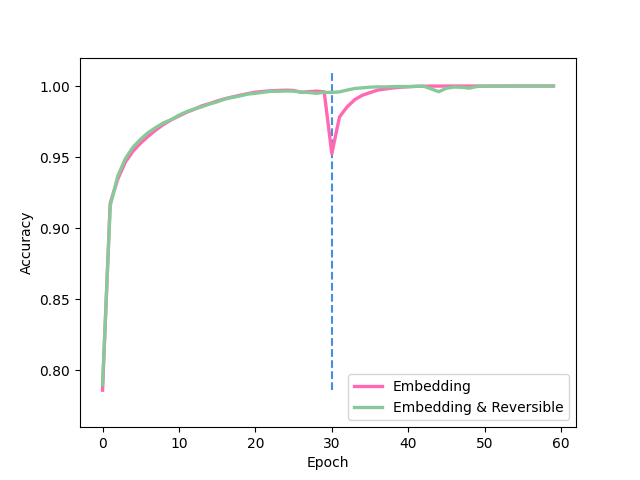}} 
	\subfigure[Group B]{\includegraphics[width=.24\textwidth]{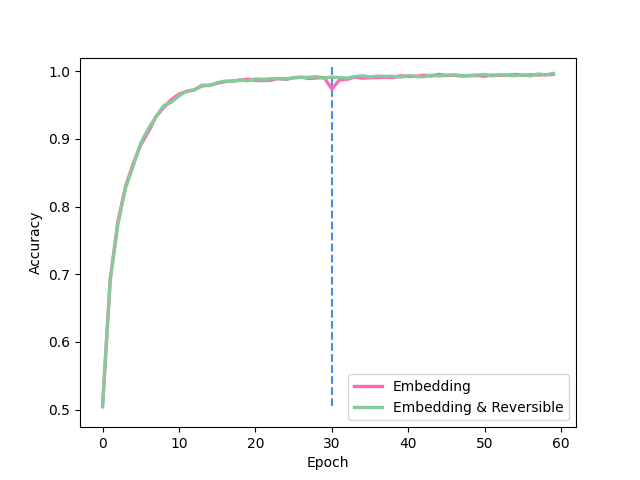}}
	
	\caption{Performance of recovering while  embedding at epoch 30.} 
	\label{Fig. Recovering Effects1}
\end{figure}

\begin{figure}[t!]
	\centering
	\subfigure[Group A]{\includegraphics[width=.5\textwidth]{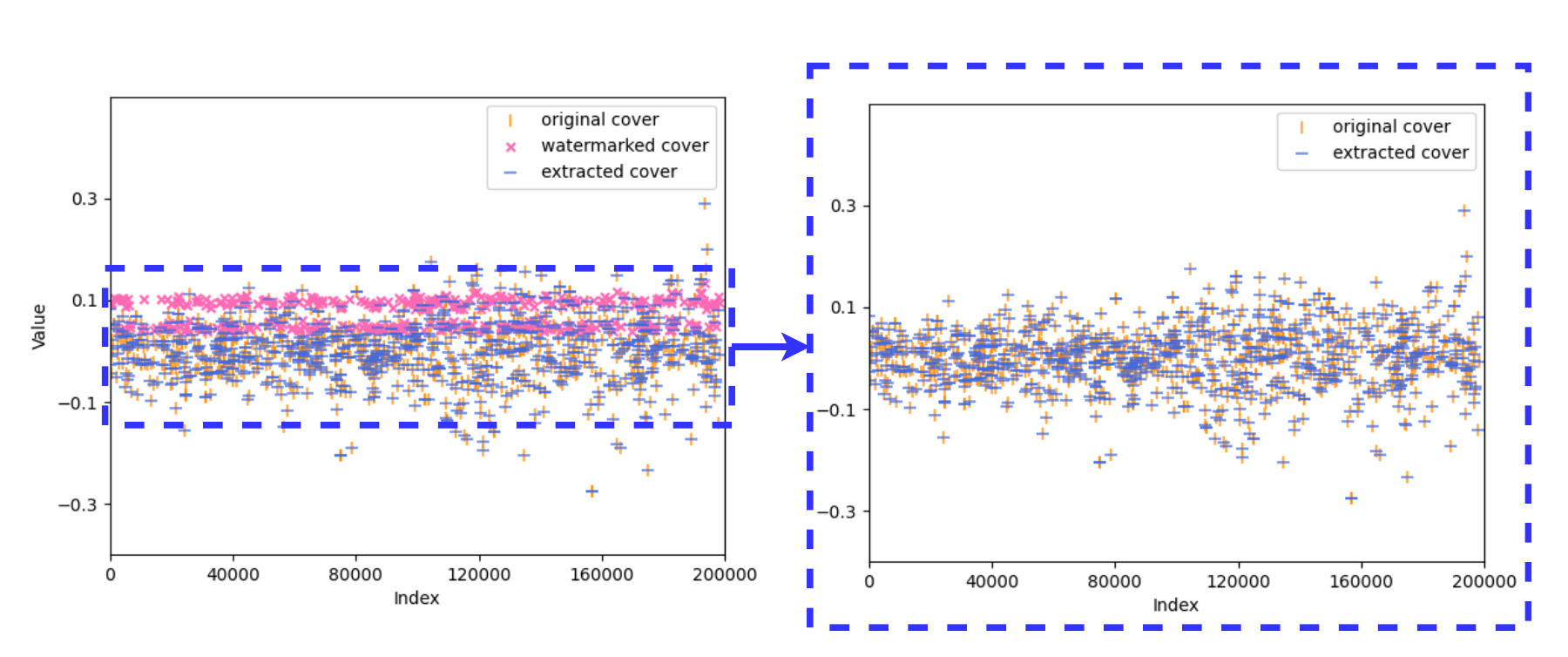}} 
	\subfigure[Group B]{\includegraphics[width=.5\textwidth]{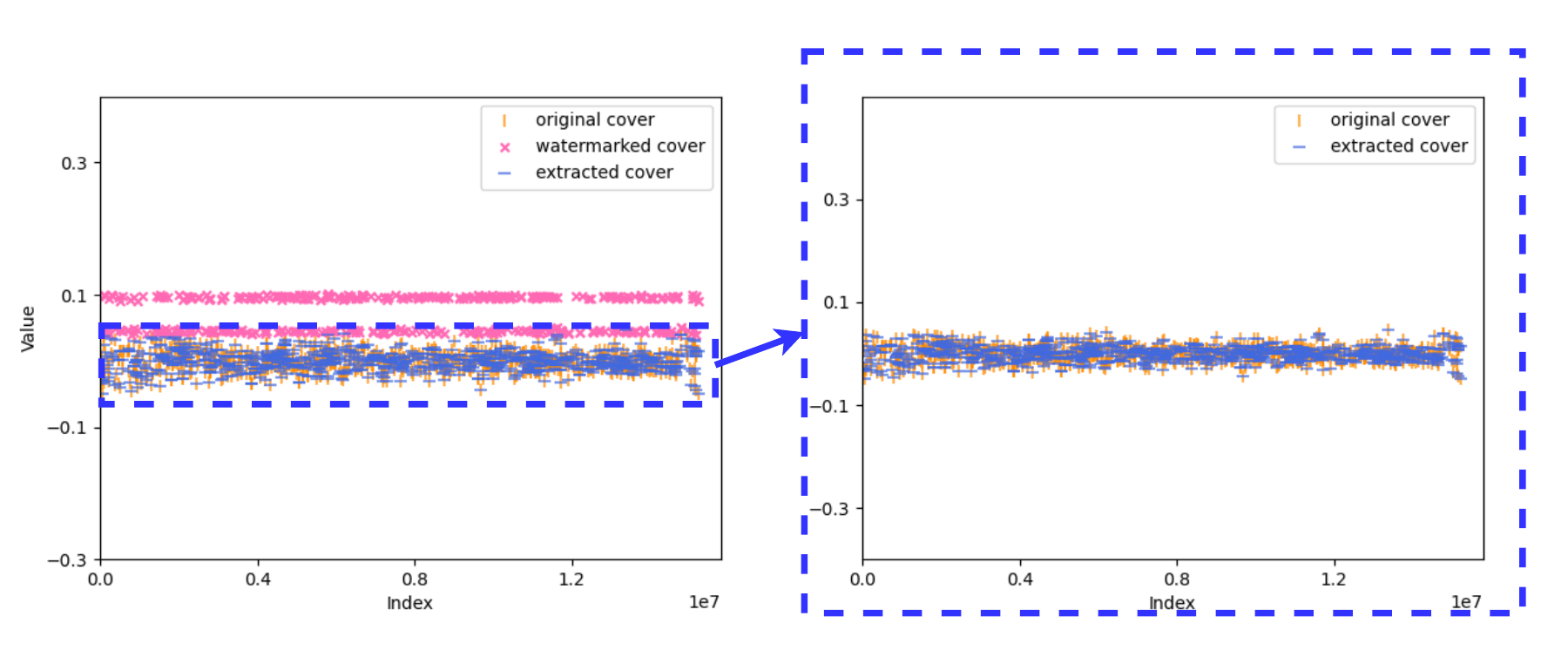}}
	\caption{The original, watermarked and recovered weights.} 
	\label{Fig. Recovering Effects2}
\end{figure}

\begin{figure}[t!]
	\centering
	\includegraphics[width=.35\textwidth]{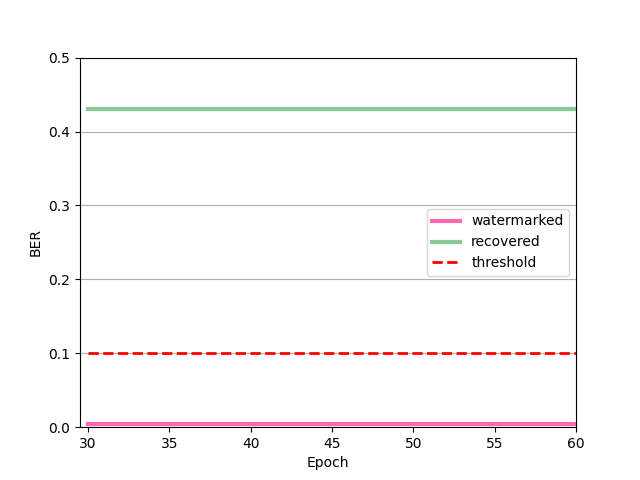}
	\caption{BER performance at different stages of recovering.} 
	\label{Fig. Recovering Effects3}
\end{figure}

\section{CONCLUSION}\label{Sec. conclusion}

In this paper, we have proposed a novel static deep neural network (DNN) watermarking scheme called Reversible QIM (R-QIM). The R-QIM scheme offers higher capacity and fidelity compared to existing methods, and it overcomes the weaknesses associated with the usability of host data under various distributions. We have also introduced two R-QIM-based schemes for integrity protection and infringement identification of DNNs.
The integrity protection scheme enables the verification of watermarked DNNs' integrity by comparing the restored model with the original model. In infringement identification, the presence of watermarks in the watermarked model can determine the legality of the current user.
Theoretical analyses and numerical simulations have demonstrated the superior performance of R-QIM compared to the method proposed in \cite{Guan_2020}. R-QIM exhibits greater flexibility in fidelity performance, higher embedding capacity, and adaptability to weights with arbitrary distributions.

In conclusion, the R-QIM scheme presents a significant advancement in DNN watermarking, offering enhanced capacity, fidelity, and applicability in various scenarios. This scheme holds promise for effective integrity protection and infringement identification of DNN models in practical applications.

%

\ifCLASSOPTIONcaptionsoff
  \newpage
\fi



%
\bibliographystyle{IEEEtranMine}
\bibliography{mybib}

%
%
%
%
%

\end{document}